\documentclass[]{elsarticle}
\usepackage{lineno,hyperref}
\usepackage{todonotes}
\usepackage{amssymb}
\usepackage{mathtools}
\usepackage{amsthm}
\usepackage{enumerate}
\usepackage[shortlabels]{enumitem}
\usepackage[]{algorithm2e}
\usepackage{caption}
\usepackage{pgf}
\usepackage{tikz}
\usetikzlibrary{arrows.meta}
\usetikzlibrary{shapes}
\usetikzlibrary{graphs}
\tikzset{
  treenode/.style = {shape=rectangle, rounded corners,
                     draw, align=center,
                     top color=white, bottom color=blue!20},
  env/.style      = {treenode, font=\ttfamily\normalsize},
  dummy/.style    = {circle,draw}
}
\newtheorem{theorem}{Theorem}
\newtheorem{lemma}{Lemma}

\newtheorem{remark}{Remark}
\newtheorem{corollary}{Corollary}
\newtheorem{example}{Example}
\newtheorem{definition}{Definition}

\newcommand{\MonF}{$\ensuremath{Mon}_F$}
\newcommand{\Var}{\mathit{Var}}

\newcommand{\SOSrule}[2]{\frac{\displaystyle #1}{\displaystyle #2}}
\newcommand{\depth}{\mathit{depth}}
\newcommand{\yes}{\mathit{yes}}
\newcommand{\no}{\mathit{no}}
\newcommand{\vend}{\mathit{end}}
\newcommand{\Act}{\mathit{Act}}
\newcommand{\pre}{\mathit{pre}}

\journal{Journal of Logical and Algebraic Methods in Programming}

\begin{document}
\begin{frontmatter}
\title{Axiomatizing recursion-free, regular monitors\tnoteref{thanks}}
\tnotetext[thanks]{This article is based on material presented at the 31st Nordic
  Workshop on Programming Theory, NWPT 2019, in Tallinn. The authors were supported by the projects `Open Problems in the Equational Logic of Processes’ (OPEL) (grant No 196050-051) and `Mode(l)s of Verification and Monitorability' (MoVeMent) (grant No~217987) of the Icelandic Research Fund, and `Runtime and Equational Verification of Concurrent Programs' (ReVoCoP) (grant No 222021), of the Reykjavik University Research Fund. Luca Aceto's work was also partially supported by the Italian MIUR  PRIN 2017 project FTXR7S IT MATTERS `Methods and Tools for Trustworthy Smart Systems'.}

\author[ru,gssi]{Luca Aceto} \ead{luca@ru.is, luca.aceto@gssi.it}
\author[ru]{Antonis Achilleos} \ead{antonios@ru.is}
\author[ru]{Elli Anastasiadi\corref{corresponding}} \ead{elli19@ru.is}
\author[ru]{Anna Ingolfsdottir} \ead{annai@ru.is}

\address[ru]{ICE-TCS, Department of Computer Science, Reykjavik University, Iceland}
\address[gssi]{Gran Sasso Science Institute, L'Aquila, Italy}

\cortext[corresponding]{Elli Anastasiadi, Menntavegur 1, 102 Reykjavik, Iceland}

%\ead[url]{www.elsevier.com}

\begin{abstract}
Monitors are a key tool in the field of runtime verification, where they are used to verify system properties by analyzing execution traces generated by processes. Work on runtime monitoring carried out in a series of papers by Aceto et al.$~$has specified monitors using a variation on the regular fragment of Milner's CCS and studied two trace-based notions of equivalence over monitors, namely verdict and $\omega$-verdict equivalence. This article is devoted to the study of the equational logic of monitors modulo those two notions of equivalence. It presents complete equational axiomatizations of verdict and $\omega$-verdict equivalence for closed and open terms over recursion-free  monitors. It is also shown that verdict equivalence has no finite equational axiomatization over open monitors when the set of actions is finite and contains at least two actions.
\end{abstract}

\begin{keyword}	
Monitors \sep Formal Verification \sep CCS \sep Equational Logic \sep Processes \sep Process Algebra \sep Axiomatization \sep Trace Equivalence \sep Verdicts.
\end{keyword}

\end{frontmatter}
%\linenumbers
\section{Introduction}
% !TeX root = axioms_Monitors_VeEq.tex
%% Last modified: apr 15 10:00:43 GMT 2021
%% Last spell checked: 

The search for equational axiomatizations of a notion of equivalence over some process description language is one of the classic topics in concurrency theory, as witnessed by the literature on this subject over the last forty years. (See, for instance,~\cite{AcetoFIL05,BaetenBR2009,BaetenB90,BergstraK84,Brookes83,Gla01,Hennessy81,HennessyM85,HoareHJMRSSSS87,Milner84,Milner:1989:CC:63446} for early references as well as survey and textbook accounts, and the papers~\cite{AcetoCFIL21,AcetoCILP20,GrabmayerF20,KappeB0WZ20} for examples of the rich body of recent contributions to this field.) This research avenue has its intellectual roots in the time-honored study of the existence of finite, (conditional) equational proof systems for equality of regular expressions, as presented in~\cite{JHConway71,Kozen94,KozenS20,redko_reg_lang,Salomaa66}.

There are manifold reasons for studying equational axiomatizations of
equivalences over processes. For example, the existence of a finite,
or at least finitely specified, equational axiomatization for some
notion of process equivalence is often considered as one of the
yardsticks to assess its mathematical tractability. Additionally,
equational axiomatizations provide a purely syntactic description of
the chosen notion of equivalence over processes and characterize the
essence of a process semantics by means of a few revealing axioms,
which can be used to compare a variety of semantics in a
model-independent way (as done, for instance,
in~\cite{Gla01}). Moreover, such axiomatizations pave the way to the
use of theorem-proving techniques to establish that two process
descriptions express the same behavior modulo the chosen notion of
behavioral equivalence~\cite{CranenGKSVWW13,GrooteR01,Lin95}, and also
play an important role in the partial evaluation of
programs~\cite{Heering86}.
 
 In this paper, we study the equational logic of the \emph{monitors}
studied by Aceto et al.~in, for
instance,~\cite{AcetoPOPL19,operGuidetoMon,FraAI17}. Monitors are a key
tool in the field of runtime verification
(see~\cite{BartocciF2018,FalconeHR13,
HavelundG05,HavelundR01a,LeuckerS09,PnueliZ06,SokolskyR12,
TabakovRV12}
and the references therein for an overview of this active research
area), where they are used to check for system properties by
analyzing execution traces generated by processes and are often
expressed using some automata-based formalism. The notion of monitorable property has been defined in a seminal paper by Pnueli and Zaks~\cite{PnueliZ06}. Intuitively, a property of finite and infinite system executions is $s$-monitorable, for some finite trace of observable events $s$, if there is an extension of $s$ after which a monitor will be able to determine conclusively whether the observed system execution satisfies or violates the property. This means that verdicts issued by monitors are \emph{irrevocable}. In that work by Pnueli and Zaks, a property is described by the set of finite and infinite executions that satisfy it. However, in the theory and practice of runtime verification, one often specifies properties finitely using formalisms such as automata or (variations on) temporal logics and studies what specifications in the chosen formalism are `monitorable’ and with what correctness guarantees—see, for instance,~\cite{BarringerRH10,monitorLTL_tLTL,sen_rosu_reg_exp_monitors}. Since monitors are part of the trusted computing base, the automated, correct-by-design \emph{monitor synthesis} from the formal specification of properties has been thoroughly studied in the literature and is often accompanied by the experimental evaluation of the overhead induced by monitoring—see, for example, the study of various approaches to the automated monitor synthesis for systemC specifications given in~\cite{TabakovRV12} and the framework for benchmarking of runtime verification tools presented in~\cite{AcetoAFI21}. 

In~\cite{AcetoPOPL19,operGuidetoMon,FraAI17}, Aceto et al.~specified monitors using a variation on the
regular fragment of Milner's CCS~\cite{M80} and studied two
trace-based notions of equivalence over monitors, namely verdict and
$\omega$-verdict equivalence. Intuitively, two monitor descriptions
are verdict equivalent when they accept and reject the same finite
execution traces of the systems they observe.
 The notion of
$\omega$-verdict equivalence is the `asymptotic version' of verdict
equivalence, in that it is solely concerned with the infinite traces
that are accepted and rejected by monitors. 
In their work, Aceto et. al. focus on determining the `monitorable’ fragment of Hennessy-Milner Logic with recursion~\cite{AcetoPOPL19,FraAI17} and provide monitor-synthesis algorithms for properties that can be expressed in that fragment. 
The key (and non-negotiable) property that the monitor synthesized from a formula $\varphi$ in the monitorable fragment of that logic should satisfy is \emph{soundness}, which means that a verdict issued by the monitor as it examines a system execution determines whether that execution satisfies $\varphi$ or not correctly. 
Naturally, sound monitors cannot produce contradictory verdicts for a given trace. 

\paragraph{Our contribution}
When monitors are described by expressions in some
monitor-specification language, such as the one employed by Aceto et
al.~in \emph{op.~cit.}, it is natural to ask oneself whether one can
(finitely) axiomatize notions of monitor equivalence over (fragments
of) that language. This study is devoted to addressing that question
in the simplest non-trivial setting. In particular, in order to stay within the realm of classic equational logic over total algebras, we consider a language that allows one to specify unsound monitors. However, all the results we present in the paper specialize to sub-languages consisting of (sound) monitors that can only issue either positive or negative verdicts.

The main results we present in this paper are complete equational characterizations of verdict equivalence over both closed (that is, variable-free) and open, recursion-free regular monitors. More specifically, we first provide an equational axiomatization of verdict equivalence over closed terms from the language of monitors we study that is finite if so is the set of actions monitors can observe (Theorem \ref{thm:GrCompFin}). The landscape of axiomatizability results for verdict equivalence over open terms turns out to be more varied. This variety is witnessed by the fact that there are three different axiomatizations, depending on whether the set of actions is infinite (Theorem~\ref{thm:CompOpenInf}), finite and containing at least two actions (Theorem~\ref{thm:CompOpenFin}) or a singleton (Theorem~\ref{thm:CompOpenUnary}). Only the axiomatization given in Theorem~\ref{thm:CompOpenUnary} is finite and we show that this is unavoidable. Indeed, verdict equivalence has no finite equational basis when the set of actions is finite and of cardinality at least two (Theorem~\ref{thm:nonfinfinal}). 

It turns out that the above-mentioned axiomatizations are also complete for $\omega$-verdict equivalence if the set of actions that monitors may observe is infinite, as in that case the two notions of equivalence coincide. On the other hand, if the set of actions is finite, $\omega$-verdict equivalence is strictly coarser than verdict equivalence. We also provide a finite, complete axiomatization of $\omega$-verdict equivalence for closed monitors in the setting of a finite set of actions (Theorem \ref{thm:GrCompOmega}). Our Theorem~\ref{thm:CompOpenOmega} gives a complete axiomatization of $\omega$-verdict equivalence over open monitors when the set of actions contains at least two actions. If the set of actions is a singleton, $\omega$-verdict equivalence has a finite equational basis (Theorem \ref{thm:compOpenOmegaUnary}). 

The equational axiomatizations we present in this article capture the
‘laws of monitor programming’~\cite{HoareHJMRSSSS87} for an admittedly
rather inexpressive language. Indeed, recursion-free regular monitors
describe essentially tree-like finite-state automata with
distinguished accept and reject states at their `leaves' with
self-loops labeled by every action. (See the operational semantics of monitors in Table~1. Note, however, that those
automata may have infinitely many transitions, if the set of actions
monitors can observe is infinite. As shown already by Milner in his
classic books on CCS~\cite{M80,Milner:1989:CC:63446}, this feature is
useful when modeling system events that carry data values. See, for
instance, the paper~\cite{BarringerGHS04} for one of the earliest
attempts to incorporate data into runtime verification.) However, as
witnessed by our results and their proofs, the study of the equational
theory of monitors modulo the notions of equivalence we consider is
non-trivial even for the minimal language studied in this paper. In
our, admittedly biased, opinion, it is therefore worthwhile to map the
territory of axiomatizability results for recursion-free regular
monitors, since results for more expressive languages will have to
build upon those we obtain in this article. We remark, in passing,
that the non-finite axiomatizability result in
Theorem~\ref{thm:nonfinfinal} is obtained over a substantially more
restrictive syntax than classic negative results for the algebra of
regular expressions, which rely on the hardness of expressing the
interplay between Kleene star and concatenation
equationally~\cite{AcetoFI98,JHConway71,redko_reg_lang}.

The contribution of this paper is entirely theoretical and we make no
claims pertaining to the applicability of our current results in the
practice of runtime verification. However, apart from their intrinsic
theoretical interest, (extensions of) the equational axiomatizations
we present might be used in the automatic, syntax-driven synthesis of
monitors from specifications of ‘monitorable properties’, as presented
in \cite{AcetoPOPL19,operGuidetoMon,FraIntroRV}, to
rewrite monitor expressions in an `equivalent, but simpler' syntactic
form, for instance by eliminating `redundant' sub-expressions. As
witnessed by the study of optimized temporal monitors for SystemC
presented in~\cite{TabakovRV12}, the investigation of monitor
optimizations based on equational rewriting or other techniques
requires a substantial experimental research effort and is outside the
scope of this article. We discuss other avenues for future research in
Section~\ref{Sect:conclusions}.

\section{Preliminaries}
% !TeX root = axioms_Monitors_VeEq.tex
%% Last modified: apr 21 10:06:14 GMT 2021

We begin by introducing recursion-free regular monitors (or simply monitors in this study) and the two notions of verdict equivalence that we study in this paper. We refer the interested reader to \cite{AcetoPOPL19,FraAI17} for background motivation and more information. 

\paragraph{Syntax of monitors}\label{sect:syntax} Let $\Act$ be a set of visible actions, ranged over by $a,b$. Following Milner~\cite{Milner:1989:CC:63446}, we use $\tau \not\in \Act$ to denote an unobservable action. The symbol  $\alpha$ ranges over $\Act \cup \{\tau\}$. Let $\Var$ be a countably infinite set of variables, ranged over by $x,y,z$. We assume that $\Act\cup \{\tau\}$ and $\Var$ are disjoint.

We write $\Act^{\omega}$ for the set of infinite sequences over
$\Act$. As usual, $\Act^*$ stands for the set of finite sequences over
$\Act$. Let $A$ be a set of finite sequences and $B$ be a set of
sequences. We write $A \cdot B$ for the concatenation of $A$ and $B$.

The collection \MonF~of (regular, recursion-free) monitors is the set of terms generated by the following grammar:

\begin{equation*}
\begin{split}
& m,n ::=~v ~~ \mid ~~a.m ~~\mid~~ m+n ~~\mid~~ x \\
& v   ::=~ \vend ~\mid~ \yes~ \mid ~\no
\end{split}
\end{equation*}
where $a \in \Act$ and $x\in \Var$. The terms $\vend$, $\yes$ and $\no$
are called \emph{verdicts}. Intuitively, $\yes$ stands for the
acceptance verdict, $\no$ denotes a rejection verdict and $\vend$ is
the inconclusive verdict, namely the state a monitor reaches when,
based on the sequence of observations it has processed so far, it
realizes that it will not be able to issue an acceptance or rejection
verdict in the future. As will be formalized by the operational
semantics of monitors to follow, verdicts are irrevocable. This means
that once a monitor reaches a verdict, it will stick to it regardless
of what further observations it makes. See, for
instance,~\cite{AcetoPOPL19,BartocciF2018,FraAI17} for a detailed
technical discussion.

Intuitively, a monitor of the form $a.m$ can observe action $a$ and behave like $m$ thereafter. On the other hand, a monitor of the form $m+n$ can behave either like $m$ or like $n$.
\begin{remark}
The work on which we build in this paper considers a setting with
three verdicts, two of which are `conclusive.' There are a number of
other approaches in the field of runtime verification that consider
many-valued verdicts. We refer the interested reader to, for instance,~\cite{BarringerFHRR12,BarringerRH10,BauerKV15,BonakdarpourFRR16,FalconeFM12} for further information. 
\end{remark}
%%%%%%%%%%%%%%%%%
Closed monitors are those that do not contain any occurrences of variables. A (closed) \emph{substitution} is a mapping $\sigma$ from variables to (closed) monitors. We write $\sigma(m)$ for the monitor that results when applying the substitution $\sigma$ to $m$. Note that $\sigma(m)$ is closed, if $\sigma$ is a closed substitution. 

\begin{definition}[\textbf{Notation}]
We use $m ~[+v]$ for a verdict $v$ to indicate that $v$ is an optional summand of $m$, that is, that the term can be either  $m$ or  $m+v$. In addition a monitor will be called $v$-free for a verdict $v$, when it does not contain any occurrences of $v$. 

For a finite index set $I = \{ i_1, \ldots , i_k\}$ and indexed set of monitors $\{m_i \}_{i \in I}$, we write $\sum_{i \in I} m_i$ to stand for $\vend$ if $I = \emptyset$ and for $m_{i_1} + \ldots + m_{i_k}$ otherwise. This notation is justified by the fact that $+$ is associative and commutative, and has $\vend$ as a neutral element, in all of the semantics we use in this paper. 
%
%Here we also introduce the generalized summation $\sum_{i \in I} m_i$ justified by the fact that $+$ is associative and commutative in all the semantics we use. $\sum_{i \in \emptyset} m_i$ stands for $\vend$. This is also consistent equationaly since $ \sum_{a \in A} a.m_a + \sum_{a \in \emptyset} a.m_a = \sum_{a \in A} a.m_a$ is always a valid equation.
\end{definition}

We now associate a notion of syntactic depth with each monitor. Intuitively, the decision a monitor $m$ takes when reading a string $s \in \Act^*$ only depends on the prefixes of $s$ whose length is at most the syntactic depth of $m$.
%From now on the $\vend$ verdict will be denoted as $\displaystyle\sum_{a \in \emptyset} a.m_a$. 
% From the way our monitors are defined one can understand that even tough they can classify arbitrary large traces, their description is finite and therefore each "important" decision will be taken before some specific constant. In order to reason about this characteristic we will need the definition of Syntactic depth. 
\begin{definition}[\textbf{Syntactic Depth}]\label{def:depth}
For any closed monitor $m \in$ \MonF, we define $\depth(m)$ as follows: \begin{itemize}
    \item $ \depth(a.m) = 1+\depth(m)$,
    \item $\depth(m_1 +m_2) = \max(\depth(m_1),\depth(m_2))$ and 
    \item $\depth(v) = 0$ for a verdict $v$. 
\end{itemize}
\end{definition}
%This notion will be useful later on when discussing the asymptotic behaviour of monitors. 

\paragraph{Semantics of monitors}

For each $\alpha\in \Act \cup \{\tau\}$, we define the transition relation $\xrightarrow[\text{}]{~\alpha~} \subseteq$ \MonF~$\times$ \MonF~as the least one that satisfies the axioms and rules in Table \ref{tab:sos_rules}.

\begin{table}
\begin{gather*}
\scalebox{1}{$ $}\, \SOSrule{}
{a.m \xrightarrow[\text{}]{a} m}
\qquad
\scalebox{1}{$ $}\, \SOSrule{m \xrightarrow[\text{}]{~\alpha~} m'}
{m + n \xrightarrow[\text{}]{~\alpha~} m'}
\qquad
\scalebox{1}{$ $}\, \SOSrule{n \xrightarrow[\text{}]{~\alpha~} n'}
{m + n \xrightarrow[\text{}]{~\alpha~} n'}
\qquad 
\scalebox{1}{$ $}\, \SOSrule{}
{v \xrightarrow[\text{}]{~\alpha~} v}
\qquad
\end{gather*}
\caption{\label{tab:sos_rules} Operational semantics of processes in \MonF.}
\end{table}

For example, $\yes + x \xrightarrow[]{\tau} \yes$ and $a.\yes + \vend \xrightarrow[]{b} \vend$, for each $a,b \in \Act$. A useful fact based on the above operational semantics is that if $m \xrightarrow[]{\tau} m'$, then $m'=v$ for some verdict $v$. 

Note that variables have no transitions. They represent under-specification in monitor behavior. For instance, monitor $a.\yes + x$ is one that we know can reach the verdict $\yes$ after having observed an $a$ action. Further information on the behavior of that monitor can only be gleaned once the variable $x$ has been instantiated via a (closed) substitution.

For $m,m'$ in \MonF~and $s=a_1\ldots a_k$ in $\Act^*,~ k\geq 0,$  $ m \xrightarrow[]{s} m'$ holds iff there are $m_0,,\ldots, m_k$ such that
 \[
m =m_0 \xrightarrow[]{a_1} m_1 \cdots m_{k-1} \xrightarrow[]{a_k} m_k = m' .
\]

Additionally, for $s\in \Act^{*}$, we use  $m \xRightarrow[]{s} m'$ to mean that: \begin{enumerate}
    \item $m~ (\xrightarrow[]{\tau})^* ~m'$ if $s=\varepsilon$, where $\varepsilon$ stands for the empty string,
    \item $m \xRightarrow[]{\varepsilon} m_1 \xrightarrow[]{a}  m_2 \xRightarrow[]{\varepsilon} m'$ for some $m_1,m_2$ if $s=a \in \Act$ and
    \item $m \xRightarrow[]{a} m_1 \xRightarrow[]{s'} m' $ for some $m_1$ if $s=a.s'$ , for some $s' \neq \varepsilon$. 
\end{enumerate}
If $m \xRightarrow[]{s} m'$ for some $m'$, we call $s$ a \emph{trace} of $m$.
% and we call $s$ a trace of $m$.

\begin{lemma}\label{lemma:axiomD}
For all $s \in \Act^*,~ m,n \in$ \MonF, and verdict $v$, $m + n \xRightarrow[]{s} v$ iff $m \xRightarrow[]{s} v$ or $n \xRightarrow[]{s} v$.
\end{lemma}
\begin{proof}
We prove both implications separately, by induction on the length of $s$. The details are straightforward and are therefore omitted. Here we limit ourselves to remarking that, in the proof of the implication from right to left, if $s = \varepsilon$ and $m = v$, say, then $v+n \xrightarrow[]{\tau} v$ by the rules in Table \ref{tab:sos_rules}. \end{proof}
%In both cases of the implication we will use induction on the length of $s$: 
%\begin{itemize}
%    \item For the implication $m \xRightarrow[]{s} v$ or $ n \xRightarrow[]{s} v$ then  $m + n \xRightarrow[]{s} v$ we proceed as follows: if $s=\varepsilon $ and $m=v$ then $v +n \xRightarrow[]{\varepsilon} v$  since $v \xRightarrow[]{\tau} v $. For $s = a.s'$ if $m \xRightarrow[]{s} v \Rightarrow m \xRightarrow[]{a} m'$ where $m' \xRightarrow[]{s'} v $ and therefore $m + n \xRightarrow[]{a} m' \Rightarrow m +n \xRightarrow[]{s} v$ 
%    \item For the implication $m + n \xRightarrow[]{s} v$ if $m \xRightarrow[]{s} v$ or $n \xRightarrow[]{s} v$ we have that: 
%
%
%If $s = \varepsilon$ and $m +n = v$ one of the $m,n$ must be equal to $v$. The case were $m +n \xRightarrow[]{s} v$ and $s = a.s'$ means that $m +n \xRightarrow[]{a} m' $ if and only if one of the $m,n$ can do an $a-$transition and arrive at $m'$. By the inductive argument either $m  \xRightarrow[]{a} m'$ which means $m \xRightarrow[]{s} v$ or, $n \xRightarrow[]{a} m'$ which means that $n \xRightarrow[]{s} v$
%\end{itemize} 

\begin{remark}
Note that the implication from right to left in Lemma \ref{lemma:axiomD} would not hold in the absence of rule $ v \xrightarrow[]{\tau} v$ in Table \ref{tab:sos_rules}. 
\end{remark}
\paragraph{Verdict and $\omega$-verdict equivalence}

Let $m$ be a (closed) monitor. We define: %\\

\begin{equation*}
\begin{split}
& L_a(m) = \{ s \in \Act^* \mid m \xRightarrow[]{s} \yes\}  ~\text{and} \\
& L_r(m) = \{ s \in \Act^* \mid m \xRightarrow[]{s} \no \} . 
\end{split}
\end{equation*}
%%%%%%%%%%%%%%%%%%%%%%%%%%
Intuitively, $L_a(m)$ denotes the set of traces that are accepted by $m$, whereas $L_r(m)$ stands for the set of traces that $m$ rejects. The sets $L_a(m)$ and $L_r(m)$ will also be referred to as the acceptance and rejection set of $m$ respectively. 
Note that we allow for monitors that may both accept and reject the same trace. This is necessary to maintain our monitors closed under $+$ and to work with classic total algebras rather than partial ones. Of course, in practice, one is interested in monitors that are consistent in their verdicts. One way to ensure consistency in monitor verdicts, which was considered in~\cite{FraAI17}, is to restrict oneself to monitors that use only one of the conclusive verdicts $\yes$ and $\no$. All the results that we present in the remainder of this paper apply to such monitors.
\begin{remark}
The reader might wonder about the connection between the languages that are accepted/rejected by recursion-free regular monitors and star-free languages~\cite{Schutzenberger65a}. A simple argument by induction on the structure of monitors shows that every recursion-free regular monitor denotes a pair of star-free languages, one for its acceptance set and one for its rejection set.  Moreover, this means that recursion-free regular monitors correspond to properties that can be expressed in LTL~\cite{Kamp1968-KAMTLA}. However, there are star-free languages (and therefore LTL properties) that cannot be described by recursion-free regular monitors. For example, the language $(ab)^*$ is  star-free (see, for instance,~\cite[page~267]{DiekertG08}) but does not correspond to any recursion-free regular monitor.

The monitors we consider in this paper output a positive or negative verdict after a finite number of computational steps, if they do so at all. This means that the linear-time temporal properties to which their acceptance and rejection set correspond  are both `Always Finitely Refutable' and `Always Finitely Satisfiable' in the sense of~\cite{PeledH18}, as proven in \cite{AcetoPOPL19}. 
 
\end{remark}

% i.e  $L_a(m) \cap L_r(m) \neq \emptyset$ is allowed.
\begin{definition}
Let $m$ and $n$ be closed monitors.
\begin{itemize}
    \item We say that $m$ and $n$ are \textbf{verdict equivalent}, written $m  \simeq n$, if $L_a(m) = L_a(n)~ $ and $~L_r(m) = L_r(n) $.
    
    \item We say that $m$ and $n$ are \textbf{$\omega$-verdict equivalent}, written $m  \simeq_{\omega} n$,   if $L_a(m) \cdot \Act^{\omega} = L_a(n) \cdot \Act^{\omega}~ $ and $ ~L_r(m) \cdot \Act^{\omega} = L_r(n) \cdot \Act^{\omega} $.
\end{itemize}
%%%%%%%%%%%%%%%%%%%%%%%%%%%%%%
For open monitors $m$ and $n$, we say that $m \simeq n$ if $\sigma (m) \simeq \sigma (n)$, for all closed substitutions $\sigma$. The relation $\simeq_{\omega}$ is extended to open monitors in similar fashion. 

\begin{example}
It is easy to see that $m + \vend \simeq m$ holds for each $m \in $\MonF . Moreover, since $L_a(\vend) = \emptyset$ and $L_r(\vend)=\emptyset$, $a.\vend  \simeq \vend$ holds for each $a \in \Act$.
\end{example}
%%A substitution $\sigma$ is closed if $\sigma(m)$ is a closed monitor 
%%$\in \MonF$ for all $x \in \Var$. 
%
%Two open monitors $m$ and $n$, will be called $\omega$-verdict equivalent if %$\sigma(m) \simeq_{\omega} \sigma(n)$, for all closed substitutions $\sigma$.
\end{definition}
%%%%%%%%%%%%%%%%%%%%%%%%%%%%%%%%%%%%%%%%%%%%
One can intuitively see that the notion of $\omega$-verdict equivalence refers to a form of asymptotic behavior. Indeed, monitors $m$ and $n$ are $\omega$-verdict equivalent if, and only if, they accept and reject the same infinite traces in the sense of \cite{AcetoPOPL19}. Next we provide a lemma that clarifies the
relations between the two notions of equivalence defined above. 

\begin{lemma}\label{lem:1}
The following statements hold: 
\begin{itemize}
    \item $\simeq$ and $\simeq_{\omega}$ are both congruences. 
    \item $\simeq \subseteq \simeq_{\omega}$ and the inclusion is strict when $\Act$ is finite. 
    \item If $\Act$ is infinite then $\simeq = \simeq_{\omega}$.
\end{itemize}
\end{lemma}
\begin{proof}
For the first claim, it suffices to prove that $\simeq$ and $\simeq_{\omega}$ are equivalence relations and that they are preserved by $a.\_$ and $+$. The proof is standard and is thus omitted.

For the second claim, the inclusion $\simeq \subseteq \simeq_{\omega}$ is easy to check using the definitions of the two relations. The fact that the inclusion is strict when the set of actions is finite follows from the validity of the equivalence 
$\yes  \simeq_{\omega} \displaystyle\sum_{a \in \Act} a.\yes $.

However, that equivalence is not valid modulo verdict equivalence since the first monitor accepts the empty string $\varepsilon$, but $\displaystyle\sum_{a \in \Act} a.\yes$ cannot.
%%%%%%%%%%%%%%%%%%%%%%%

Finally, suppose that $\Act$ is infinite. Assume that $m$ and $n$ are $\omega$-verdict equivalent and that $s$ is a finite trace accepted by $m$. We will argue that $n$ also accepts $s$. To this end, note that, since $\Act$ is infinite, there is some action $a$ that does not occur in $m$ and $n$. Since $m$ accepts $s$, the infinite trace $sa^{\omega}$ is in $L_a(m)\cdot \Act^{\omega}$. By the assumption that $m$ and $n$ are $\omega$-verdict equivalent, we have that $sa^{\omega}$ is in $L_a(n)\cdot \Act^{\omega}$. As $a$ does not occur in $n$, it is not hard to see that $n$ accepts $s$. Therefore, by symmetry, $m$ and $n$ accept the same traces. The same argument shows that $L_r(m)=L_r(n)$, and therefore $m \simeq n$. \end{proof}

\paragraph{Equational logic}
An axiom system $\mathcal{E}$ over \MonF~is a collection of equations $m = n$ expressed in the syntax of \MonF. An equation $m = n$  is derivable from an axiom system $\mathcal{E}$ (notation $\mathcal{E} \vdash m = n $) if it can
be proven from the axioms in $\mathcal{E}$ using the rules of equational logic (reflexivity, symmetry, transitivity, substitution and closure under the \MonF~contexts). See Table~\ref{tab:EL}. In the rest of this work we shall always implicitly assume, without loss of generality, that equational axiom systems are closed with respect to symmetry, i.e., that if $m = n$ is an axiom, so is $n =m$. 

\begin{table}%[h]
\centering
\begin{tabular}{|ll|} 
\hline
Reflexivity    &                                                                                                                                                                       \\                                                             & $t = t$                                                                                                                                                          \\
Symmetry                                                             &                                                                                                                                                                             \\                                                            & $\SOSrule{t =t'}{t'= t}$ \\
Transitivity  &                                                                                                                                                                             \\
 & $\SOSrule{t_1 = t_2, ~t_2 = t_3}{t_1= t_3}$                                                                                         \\
Congruence (For any $n$-ary $f$) &                                                                                                                                                                             \\ & $\SOSrule{t_i = t_i',~ i = 1,2,\ldots ,n}{f(t_1,\ldots t_n ) = f(t_1',\ldots, t_n')}$  \\

Substitutivity (For each substitution $\sigma$)                 &                                                                                                                                                                             \\                                                                 & $\SOSrule{t = t'}{\sigma(t) = \sigma(t')}$                                                                                  \\
\hline
\end{tabular}
\caption{Rules of equational logic \label{tab:EL}}

\end{table}

We say that $\mathcal{E}$ is \emph{sound} with respect to $\simeq$ when $m \simeq n$ holds whenever $\mathcal{E} \vdash m = n$. We say that $\mathcal{E}$ is \emph{complete}  with respect to $\simeq$ when $\mathcal{E}$ can prove all the valid equations $m \simeq n$. Similar definitions apply for $\omega$-verdict equivalence. 
The notion of completeness, when limited to closed terms, is referred to as \textit{ground completeness}.

\section{A ground-complete axiomatization of verdict and $\omega$-verdict equivalence}
% !TeX root = axioms_Monitors_VeEq.tex
%%% Last modified: apr 21 09:35:51 GMT 2021
%%% Last spell checked:

Our goal in this paper is to study the equational theory of $\simeq$ and $\simeq_{\omega}$ over \MonF. Our first main result is to give a ground-complete axiomatization of verdict equivalence over \MonF. To this end, consider the axiom system $\mathcal{E}_v$, whose axioms are listed in Table \ref{groundaxioms}.

%Our proposed axioms system for verdict equivalence is $\mathcal{E}_{v}$, whose axioms are: 
\begin{table}[h]
\begin{minipage}{0.5\textwidth}
\begin{equation*}
\begin{split}
& \textbf{(A1)} ~x + y  = y + x \\
& \textbf{(A2)} ~x + (y + z)  = (x+y)+z \\
& \textbf{(A3)} ~x + x  = x \\
& \textbf{(A4)} ~x + \vend  = x 
\end{split}
\end{equation*}
\end{minipage}
\begin{minipage}{0.46\textwidth}
\begin{equation*}
\begin{split}
& \mathbf{(E_a)}~ a.\vend = \vend ~(a \in \Act)\\
& \mathbf{(Y_a)}~  \yes = \yes + a.\yes ~(a \in \Act) \\
& \mathbf{(N_a)}~  \no = \no +  a.\no ~(a \in \Act) \\
& \mathbf{(D_a)}~ a.(x + y)  = a.x +a.y ~(a \in \Act)
\end{split}
\end{equation*}
\end{minipage}
\caption{The axioms of $\mathcal{E}_v$}
\label{groundaxioms}
\end{table}

\begin{remark} Note that $\mathcal{E}_v$ is finite, if so is $\Act$.
\end{remark}

The subscript $_v$ in the naming scheme of the axiom set refers to the kind of equivalence that it axiomatizes, namely verdict equivalence. It will later be replaced with $_\omega$ when we study $\omega$-verdict equivalence and used accordingly from that point forward.

We provide now the following lemma as an observation on the number of necessary axioms when $\Act$ is finite and as an example proof based on these axioms. 
\begin{lemma} When $\Act$ is finite, the family of axioms $(Y_a)$ can be replaced with 
$$ \mathbf{(Y)} ~ \yes = \yes + \displaystyle\sum_{a \in \Act} a.\yes \text{.}$$ Similarly the family of axioms $(N_a)$ can be replaced with $$ \mathbf{(N)}~ \no = \no + \displaystyle\sum_{a \in \Act} a.\no\text{.}$$
\end{lemma}
\begin{proof}
It is not hard to see that the equation $Y$ can be proved by using the family of equations $Y_a$. For the converse we can use axioms $A3$ and $Y$ to prove any equation $\yes = \yes +b.\yes$ of the family $\{Y_a \mid a  \in \Act \} $. Indeed, $\mathcal{E}_v$ proves  $$\yes = \yes + \displaystyle\sum_{a \in \Act} a.\yes = \yes + \displaystyle\sum_{a \in \Act} a.\yes + b.\yes  = \yes + b.\yes .$$ \end{proof}
%Throughout this work in most cases we will use the axioms $Y_a$ and $N_a$ even when $\Act$ is finite.
\begin{theorem}\label{thm:GroundSound}
$\mathcal{E}_{v}$ is sound modulo $\simeq$. That is, if $\mathcal{E}_{v} \vdash m = n$ then $m \simeq n$, for all $m,n \in $\MonF.

\end{theorem}

\begin{proof} It suffices to prove soundness for each of the axioms separately. The details of the proof are standard and therefore omitted. \end{proof}

%
%\begin{itemize}
%    \item \textbf{$A1-4$} are proved easily since both sides of the equations correspond to identical expressions over \MonF.
%    \item For \textbf{$E_{a}$} it is enough to argue that none of the sides of the equation interfere in any way with the verdicts of the relevant monitors. 
%    \item \textbf{$Y_a$} (similarly for \textbf{$N_a$}) For any occurrence of a$\yes$ in a monitor $m \in$ \MonF we have that every sequence leading to that $\yes$ will also be a member of $L_a(m)$ when we replace it with a $\yes + a.\yes$. The converse also holds since a verdicts are irrevocable and therefore all sequences that lead to a $\yes$ will also lead to a$\yes$ when extended by any action of $\Act$.
%    \item \textbf{$D_a$} is easily proved by the use of Lemma \ref{leamma:axiomD} $$L_a(a.(m+n)) = \{a.s \mid s\in L_a(m+n)\} = \{ a.s \mid s \in L_a(m) \cup L_a(n) \}= L_a(a.m + a.n)$$ .
%    
    % $$L_a\{a.(x+y)\} = \{s: \Act^* ~\mid~ x+y \xRightarrow[]{a.s} \yes \} = \{s: \Act^* ~\mid~ x \xRightarrow[]{a.s} \yes \vee   y \xRightarrow[]{a.s} \yes \}$$ $$ = \{s: \Act^* ~\mid~ x \xRightarrow[]{a.s} \yes \}\cup \{s: \Act^* ~\mid~ y \xRightarrow[]{a.s} \yes \} = L_a\{a.x\}\cup L_a\{ a.y\}  = L_a\{a.x+ a.y\}$$
%\end{itemize}
% This completes the proof of the soundness of $\mathcal{E}_{v}$. 

% This means that all monitors that can be proved equal thought the equations of $E_veq$ will also be verdict equivalent. 

In what follows, we will consider terms up to axioms $A1$-$A4$.

A fact that will be proven useful later on is the following: If $ m \xrightarrow[]{a} n$ then $A1-A4,E_a,Y_a,N_a ~\vdash m = m + a.n.$ This follows easily by induction on the size of $m$ and a case analysis on its form and it is thus omitted.  
%Proof: The proof is by structural induction on m. We proceed by a case analysis on the form of m.
%
%    If m = \vend then \vend  = \vend + \vend = \vend + a.\vend.
%    If m = \yes then use A3 and Y_a.
%    If m = \no then use A3 and N_a.
%    If m = a.n then use A3.
%    If m = m_1 + m_2 use induction and A1-A4. 

We will now prove that the axiom system $\mathcal{E}_{v}$ is ground complete for verdict equivalence. 
\begin{theorem}\label{thm:GrCompFin}
$\mathcal{E}_{v}$ is ground complete for $\simeq$ over \MonF. That is, if $m,n$ are closed monitors in \MonF~and $m \simeq n$ then $\mathcal{E}_{v} \vdash m = n$.                    \end{theorem}
%%%%%%%%%%%%%%%%%%
As a first step towards proving that $\mathcal{E}_{v}$ is complete over closed terms, we isolate a notion of normal form for monitors and prove that each closed monitor in \MonF~can be proved equal to a normal form using the equations in $\mathcal{E}_{v}$.
% It remains to be seen whether each possible verdict equivalence is provable by $E_{veq}$. For this proof we will use structural induction over the form of the monitors $m,n \in \MonF$. In order to be able to argue over the structure of terms in $\MonF$ we will first need some new terminology. Over the syntax given in \ref{sect:syntax}  we can define the notion of a \textit{normal from}. 

\begin{definition}\textbf{(Normal Form)}
A normal form is a closed term $m \in$ \MonF~of the form: 
$$\displaystyle\sum_{a \in A} a.m_a~ [+ \yes]~ [+ \no]$$
for some finite $A \subseteq \Act$,
%and $\{ m_a \mid a \in A   \}$,
where each $m_a$ is a term in normal form that is different from $\vend$.
\end{definition}
Note that, by taking $A=\emptyset$ in the definition above, we obtain
that $\vend$ is a normal form. In fact, it is the normal form with the
smallest size.

\begin{lemma}\label{lem:endVerdict}
The only normal form that does not contain occurrences of $\yes$ and $\no$ is $\vend$.
\end{lemma}

\begin{proof}

We proceed by induction on the size of a normal form $m$. Our base case is a verdict $v$. The only such verdict that does not contain an occurrence of either $\yes$ or $\no$ is $\vend$, which trivially satisfies the lemma. Assume now that $m = \displaystyle\sum_{a \in A} a.m_a$ is a normal form satisfying the statement of the lemma. Since each $m_a$ is $\yes$- and $\no$-free, by inductive hypothesis, $m_a = \vend$. This is only possible if $A = \emptyset$. Thus $m= \vend$. 
\end{proof}

%If $m = \vend$ then it is in normal form and it satisfies this lemma. If $m$ is some $\displaystyle\sum_{a \in A} a.m_a$ for some finite $A \subseteq \Act$ and $\{ m_a \mid a \in A   \}$. For each $a \in A$ the $m_a$ monitor does not contain any `$\yes$ or $\no$. Therefore by induction hypothesis we can say that $\forall a \in A, m_a = \vend$. Therefore $m = \displaystyle\sum_{a \in A} a.\vend$. However a normal form is not allowed to perform transitions that result to $\vend$. For $m$ this means that it is allowed no transitions since we have showed that if it was allowed any they would result in $\vend$. Therefore $m = \displaystyle\sum_{a \in \emptyset} a.m_a = \vend$.  

\begin{lemma}\textbf{(Normalization)}\label{lem:normClosed}
Each closed term $m \in $\MonF~is provably equal to some normal form $m'$ with $\depth(m') \leq \depth(m)$. 
%using axioms axioms $A1-4,~ E_a~(a \in \Act)$ and $D_a~ a \in \Act$
\end{lemma}

\begin{proof} We prove the claim by induction on the lexicographic ordering $\prec$ over pairs $(\depth(m),size(m))$ of a monitor $m$, where $size(m)$ denotes the length of $m$ in symbols. We proceed with a case analysis on the form $m$ may have. Our induction basis will be a verdict $v$. If $v = \vend$ then the monitor is already in normal form. Otherwise: 
If $m = v$ for some verdict $v =\yes$ or $\no$ then it is proved equal to $v+ \vend$ (from axiom $A4$). Indeed this normal form of a non-$\vend$ verdict has depth less or equal to that of the initial monitor. 

Our induction hypothesis is that, for all monitors $m_0 \in$ \MonF~up such that $(\depth(m_0),size(m_0)) \prec (\depth(m), size(m))$, we have that $\mathcal{E}_{v} \vdash m_0 = m_0'$ with $m_0'$ in normal form and $\depth(m_0') \leq \depth(m_0)$. 

Assume that $m = a.n$ then clearly $n$ has depth less than that of $m$ and therefore by the inductive hypothesis $\mathcal{E} \vdash n \simeq n'$ where $n'$ is in normal form and of depth less or equal than $n$. If $n' = \vend$ then $\mathcal{E}_v \vdash m = \vend$ (using $E_a$) which is a normal form of smaller depth. Otherwise $a.n'$ is also a normal form. 

Assume that $m = m_1 + m_2$. 
% In this case we have that the monitors $m_1,m_2$ might have depth equal to that of $m$. However they both  have size strictly smaller than $m$ and therefore even in the case of equal depth they appear earlier in the lexicographic ordering of pairs we have chosen to apply induction on. 
 By applying the induction hypothesis we have that $$\mathcal{E}_v \vdash m_1 = \sum_{a \in A_1} a.m_{1a}~ [+\yes][+\no] ~\text{and}~ \mathcal{E}_v \vdash m_2 = \sum_{a \in A_2} a.m_{2a}~ [+\yes][+\no].$$  Therefore by applying axioms from $\mathcal{E}_v$ we can rewrite $m$ as: 
    
    $$m =\displaystyle\sum_{a \in A_1\setminus A_2} a.m_{1a}' + \sum_{a \in A_2\setminus A_1} a.m_{2a}' + \sum_{b \in A_1\cap A_2} b.(m_{1b}' + m_{2b}') ~[+\yes][+\no].$$ 

Where by the statement of the lemma we have:  
$$depth\left(\displaystyle\sum_{a \in A_1\setminus A_2} a.m_{1a}'\right) \leq depth (m_1) \leq \depth(m)$$
 and similarly:
 $$\depth\left(\displaystyle\sum_{a \in A_2\setminus A_1} a.m_{2a}'\right) \leq depth (m_2) \leq \depth(m).$$ 
 It remains to show that the summand $\displaystyle\sum_{b \in A_1\cap A_2} b.(m_{1b}' + m_{2b}')$ is equal to a normal form and that it has depth less or equal to that of $m$. However, this is not trivial to see, since the terms $m_{1a}'$ and $m_{2b}'$ have been rewritten by the normalization procedure and therefore we cannot guarantee that their summation has size less of that of $m$ (applying the inductive hypothesis only results in terms of smaller depth but not size as we saw for instance in the case of normalization of verdicts). However we have the following: $$\depth( m_{1b}' + m_{2b}') = max[\depth(m_{1b}'),\depth(m_{2b}')] $$ $$ <~ max[\depth(m_1'),\depth(m_2')].$$

 The later of the above quantities is guaranteed to be less than or equal to $\depth(m)$ by the inductive hypothesis. Therefore we still have that the monitor $m_{1b}' + m_{2b}'$ appears earlier in the lexicographic ordering and therefore $\mathcal{E}_v$ can prove it equal to a normal form of smaller depth. We will call this normal form $m_b'$. We have therefore that that $\depth(m_b') \leq depth (m_{1b}' + m_{2b}') $.  We now have the necessary result that $$\mathcal{E}_v \vdash m = m' = \displaystyle\sum_{a \in A_1 \cup A_2} a.m_a ~ [+\yes][+\no],$$ where each $m_a$ is in normal form and of depth strictly less than that of $m$ which means that $\depth(m') \leq \depth(m)$ and we are done. \end{proof}

Since now we have that each term in \MonF~is provably equal to a normal form, we might attempt to prove Theorem \ref{thm:GrCompFin} by arguing that the normal forms of two verdict equivalent monitors are identical. However, it turns out that this is not true. Consider, for example, the case were $m = \yes$ and $n = \yes + a.a.a.\yes$. These two monitors are clearly verdict equivalent as $L_a(m)=L_a(n) = \Act^*$ and $L_r(m)=L_r(n) = \emptyset$. However, even though they are in normal form they are not syntactically equal. Intuitively, $a.a.a.\yes$ in monitor $n$ is redundant, as it can be absorbed by  $\yes$.  In what follows, we will show how to reduce the normal form of a monitor further using equations in $\mathcal{E}_v$ in order to eliminate such redundant sub-terms.

\begin{lemma}\label{lem:ReducedNF}  The following statements hold for any monitor in \MonF: 
\begin{enumerate}
    \item For each action $a$, if $m$ is a closed $\no$-free term then  $\mathcal{E}_{v} \vdash \yes + a.m = \yes$.
    \item For each action $a$, if $m$ is a closed monitor that contains occurrences of both $\yes$ and $\no$ then $\mathcal{E}_{v} \vdash \yes + a.m = \yes + a.n$ for some $\yes$-free closed monitor n. 
    \item For each action $a$, if $m$ is a closed $\yes$-free term then  $\mathcal{E}_{v} \vdash \no + a.m = \no$.
    \item For each action $a$, if $m$ is a closed monitor that contains occurrences of both $\yes$ and $\no$ then $\mathcal{E}_{v} \vdash \no + a.m = \no + a.n$ for some $\no$-free closed monitor n. 
\end{enumerate}
  
\end{lemma}

\begin{proof}
We only prove statements \textit{1} and \textit{2} as the proofs of \textit{3} and \textit{4} are similar. We will use structural induction on $m$.
\begin{enumerate} 
    \item If $m$ is a verdict other than $\no$ then the claim follows using axioms $E_a,~Y_a$ and $A4$ appropriately. If $m =b.m'$ where $m'$ is $\no$-free then $\mathcal{E}_v$ derives: $$\yes +a.m \overset{\mathrm{Y_a}}{=}  \yes + a.\yes + a.m \overset{\mathrm{D_a}}{=} \yes + a.(\yes + b.m')  \overset{\mathrm{\textbf{I.H.}}}{=} \yes +a.\yes \overset{\mathrm{Y_a}}{=} \yes .$$ If $m$ is of the form $m_1 +m_2$ where $m_1,m_2$ are $\no$-free, then it suffices to apply axiom $D_a$ and the induction hypothesis.
    \item Assume that $m$ contains occurrences of both $\yes$ and $\no$. We will show that $\mathcal{E}_v \vdash \yes + a.m = \yes + a.n$ for some $\yes$-free monitor $n$. 
    
   If $m=v$ for some verdict $v$ then the claim follows vacuously.

   If $m = b.m'$ for some $m'$ that contains both $\yes$ and $\no$ then there is some $\yes$-free $n'$, such that $n = b.n'$, and $\mathcal{E}_v$ derives : 
    $$ \yes + a.m = \yes + a.b.m' \overset{\mathrm{Y_a}}{=} \yes + a.\yes + a.b.m' \overset{\mathrm{D_a}}{=} \yes + a.(\yes + b.m')$$ 
   
     and for some $\yes$-free $n'$ s.t. $n = b.n'$: 
   $$  \overset{\mathrm{I.H}}{=} \yes + a.(\yes + b.n') \overset{\mathrm{D_a}}{=} \yes + a.\yes + a.b.n' \overset{\mathrm{Y_a}}{=} \yes +a.n .$$

Finally if $m = m_1 + m_2$ then $\mathcal{E}_v$ can derive:

$$\yes + a.(m_1 + m_2) \overset{\mathrm{D_a}}{=} \yes + a.m_1 + a.m_2.$$ We now isolate the following cases based on what verdicts the monitors $m_i,~i \in \{ 1,2\}$ contain. If any $m_i, ~i \in \{ 1,2\}$ is both $\yes$- and $\no$-free it must be equal to $\vend$ as it is in normal form and therefore $\mathcal{E}_v \vdash \yes + a.m_i = \yes$. If $m_i, ~i \in \{ 1,2\}$ contains occurrences of both $\yes$ and $\no$, then the induction hypothesis yields that $$\mathcal{E}_v \vdash \yes + a.m_i = \yes + a.n_i$$ for some $\yes$-free $n_i$. If $m_i, i \in \{ 1,2\}$ is $\yes$-free we already have the result that $\mathcal{E}_v \vdash \yes + a.m_i = \yes + a.n_i$ for some $\yes$-free monitor $n_i$ (which in this case coincides with $m_i$). Finally, if some $m_i$ is $\no$-free then, by statement \textit{1} in the lemma, 

$$ \mathcal{E}_v \vdash \yes + a.m_i = \yes.$$

Combining these observations, we have that:
$$ \mathcal{E}_v \vdash \yes + a.m = \yes  +a.n_1 + a.n_2~~$$ where both $n_1$ and $n_2$ are $\yes$-free and therefore by axiom $D_a$: 
$$ \mathcal{E}_v \vdash \yes+ a.m = \yes + a.n$$ for some $\yes$-free monitor $n$.

%    
%    If $m = \yes +\no$ ($m$ has to contain occurrences of both $\yes$ and $\no$) then by $A3$ we have the requested result. We start by defining as $m'$ the normal form of $m$. We can rewrite $m'$ as $\displaystyle\sum_{a_1 \in A_1} a_1.m_{a_1} + \sum_{a_2 \in A_2} a_2.m_{a_2}$ where $m_{a_1}$ is $\no$-free and $m_{a_2}$ is $\yes$-free. This form can be constructed by utilizing axiom $D_a$. Note that there might be actions $a$ in both $A_1$ and $A_2$ which means that this is not a normal form. By the previous case of this lemma this reduces to :$$ \displaystyle\sum_{a_2 \in A_2} a_2.m_{a_2} + \yes$$ and each $m_{a_2}$ is $\yes$-free and not equal to $\vend$.
\end{enumerate}
\end{proof}
% It suffices to prove that any sequence $s \in L_a(m)$ i.e any occurrence of ```$\yes'- in $m$ (similarly for occurrences of``$\no'--) can be found in the right hand side  $\yes$. We are presenting an inductive argument over the size of the sequence.  

% For $s = (s1)$ we can manipulate the right hand side ```$\yes'- verdict ,by utilising axioms $Y_a, A3, D_a$ in the following manner: $$\yes = \yes + a.\yes = \yes + a.(\yes+s1.\yes)$$ $$
%  =^{(D_{s1})} \textbf{\yes + a.\yes} + a.s1.\yes = \yes + a.s1.\yes$$
% Which then we can repeat inductively for any sequence size. 

The above lemma suggests the notion of a \textit{reduced} normal form. 

\begin{definition}\label{ReducedNF_Closed}\textbf{(Reduced normal form) }
A reduced normal form is a term $$m  = \displaystyle\sum_{a \in A} a.m_a ~[+\yes] ~[+\no]$$ in normal form, where if $v \in \{\yes,\no\}$ is a summand of $m$ then each $m_a$ is $v$-free and in reduced normal form.
\end{definition}
\begin{remark}
Note here that if $\displaystyle\sum_{a \in A} a.m_a +\yes +\no$ is in reduced normal form then $A = \emptyset$.
\end{remark}
\begin{lemma}\label{lem:realReducedNF_Closed}
Each monitor in normal form is provably equal to a monitor in reduced normal form.
\end{lemma}
\begin{proof}
The claim follows from Lemma \ref{lem:ReducedNF}, using induction on the depth of the normal form. \end{proof}

We are now ready to complete the proof of Theorem \ref{thm:GrCompFin}.

\begin{proof}[Proof of Theorem \ref{thm:GrCompFin}]\label{thm:ProofOfMainThm}
Since each monitor is provably equal to a reduced normal form (Lemma \ref{lem:realReducedNF_Closed}), and by the soundness of $\mathcal{E}_v$ (Theorem \ref{thm:GroundSound}), it suffices to prove the claim for verdict equivalent reduced normal forms $m$ and $n$. We proceed by induction on the sum of the sizes of $m$ and $n$, and a case analysis on the possible form $m$ may have. 
%Recall that we consider terms up to $A4$ and we thus remove $end$ summands in reduced normal forms. 
\begin{enumerate}
    \item Assume that $m = \yes +\no \simeq n$. Since $L_a(m) = L_r(m) = \Act^*$, it follows that $n$ has both $\yes$ and $\no$ as summands. Since $n$ is in reduced normal form it must be the case that $n = \yes +\no$, and we are done.
    \item Assume that $m = \displaystyle\sum_{a \in A}a.m_a +\yes \simeq n$, where, for all $a \in A,~ m_a$ is $\yes$-free and in reduced normal form and $n = \displaystyle\sum_{b \in B} b.n_b [+\yes] [+\no]$, where each $n_b$ is in reduced normal form and is $v$-free, if $v$ is a summand of $n$. \label{case:4}
Since $\varepsilon \in L_a(m)\setminus L_r(m)$, we have that $\yes$ is a summand of $n$ and $\no$ is not. Thus $n = \displaystyle\sum_{b \in B}b.n_b +\yes$, and each $n_b$ is $\yes$-free. We claim that: 
    \begin{enumerate}[start=1,label={(\bfseries C\arabic*)}]

        \item $A=B$ and \label{claim:A=B}
        \item for all $a \in A,~ m_a \simeq n_a$.
    \end{enumerate}
    To prove that $A=B$, we assume that $a \in A$. Since $m_a$ is $\yes$-free and different from $\vend$, there is some $s \in \Act^*$ such that $a.s \in L_r(m)$. As $m \simeq n$, we have that $a.s \in L_r(n)$. We conclude that $a \in B$ and $s \in L_r(n_a)$. By symmetry, claim \ref{claim:A=B} follows.

    We now show that $m_a \simeq n_a$ for each $a \in A$. Since $m_a$ and $n_a$ are $\yes$-free, $L_a(m_a) = L_a(n_a) = \emptyset$. We pick now some arbitrary $s \in L_r(m_a)$ ($L_r(m_a) \neq \emptyset$ because $m_a \neq \vend$). This means that $a.s \in L_r(m) = L_r(n)$ and therefore $s \in L_r(n_a)$. The claim follows by symmetry. By the induction hypothesis, $\mathcal{E}_{v} \vdash m_a= n_a$ for each $a \in A = B$. Therefore
     $$m = \displaystyle\sum_{a \in A}a.m_a +\yes = \displaystyle\sum_{b \in B}b.n_b +\yes  = n$$ 
     is provable from $\mathcal{E}_{v}$ and we are done. 
\item We are left with the case where $m = \displaystyle\sum_{a \in A}a.m_a +\no \simeq n$ and the case $m = \displaystyle\sum_{a \in A}a.m_a$. The proofs for those cases are similar to the one for case \textit{2} and are thus omitted. 
\qedhere
\end{enumerate}
\end{proof}

\subsection{Axiomatizing $\omega$-verdict equivalence} When $\Act$ is infinite, by Lemma~\ref{lem:1} and Theorem~\ref{thm:GrCompFin}, $\mathcal{E}_v$ gives a ground-complete axiomatization of $\omega$-verdict equivalence as well. However, when $\Act$ is finite, $\mathcal{E}_{v}$ is not powerful enough to prove all the equalities between closed terms that are valid with respect to $\omega$-verdict equivalence. The new axioms needed to achieve a ground complete axiomatization in this setting are:  \\
\begin{minipage}[c]{0.5\textwidth}
\centering
   $\mathbf{(Y_{\omega})}$ ~$\yes = \displaystyle\sum_{a \in \Act} a.\yes$

\end{minipage} 
\begin{minipage}[c]{0.5\textwidth}
\centering
    $\mathbf{(N_{\omega})}$~ $\no = \displaystyle\sum_{a \in \Act} a.\no$.

\end{minipage} 
The resulting axiom system is called $\mathcal{E}_{\omega}$. 
\begin{remark}
The soundness of the new axioms is trivially shown since $$L_a(\yes)\cdot \Act^{\omega} =\Act^* \cdot \Act^{\omega} = \Act^+\cdot \Act^{\omega} = L_a(\displaystyle\sum_{a \in \Act} a.\yes)\cdot \Act^{\omega}$$ while $L_r(\yes) =L_r(\displaystyle\sum_{a \in \Act} a.\yes) = \emptyset$ (and symmetrically for the $N_{\omega}$ equation).
\end{remark}
\begin{theorem}\label{thm:GrCompOmega}
$\mathcal{E}_{\omega}$ is ground complete for $\simeq_{\omega}$ over closed terms when $\Act$ is finite. That is if $m,n$ are closed monitors in \MonF~and $m \simeq_{\omega} n$ then $\mathcal{E}_{\omega} \vdash m = n$.
\end{theorem}

\begin{proof} 

By Lemma \ref{lem:realReducedNF_Closed} we may assume that $m$ and $n$ are in reduced normal form. We will prove the claim by induction on the sizes of $m$ and $n$ for two $\omega$-verdict equivalent monitors $m,n$ in reduced normal form. 

We will proceed by a case analysis of the form $m$ may have and limit ourselves to presenting the proof for a few selected cases that did not arise in the proof of Theorem \ref{thm:GrCompFin}. 
\begin{itemize}
    \item Assume that $m = \yes +\no \simeq_{\omega} \displaystyle\sum_{a \in A}a.n_a = n$. First of all note that $A= \Act$. Indeed if $a \in \Act\setminus A$ then $a^{\omega} \in (L_a(m)\cdot \Act^{\omega})\setminus (L_a(n)\cdot \Act^{\omega})$ which contradicts our assumption that $m \simeq_{\omega} n$. Moreover, it is not hard to see that, for each $a \in \Act$, $L_a(n_a)\cdot \Act^{\omega} = L_r(n_a)\cdot \Act^{\omega} = \Act^{\omega}$.
    This means that, for each $a \in \Act$, $n_a \simeq_{\omega} \yes +\no$. By induction, for each $a \in \Act$, we have that $\mathcal{E}_{\omega} \vdash n_a = \yes +\no$. Thus, $\mathcal{E}_{\omega} \vdash n = \displaystyle\sum_{a \in \Act} a.(\yes +\no) $. From axiom $D_a$, $\mathcal{E}_{\omega} \vdash n = \displaystyle\sum_{a \in \Act} a.\yes + \sum_{a\in \Act} a.\no $ which from our two new axioms $Y_{\omega},N_{\omega}$ yields  $\mathcal{E}_{\omega} \vdash n = \yes +\no = m$, and we are done.  
    \item Assume that $m = \yes +\no \simeq_{\omega} \displaystyle\sum_{a \in A}a.n_a +\yes$, with each $n_a$ being $\yes$-free and different from $\vend$. Again, reasoning as in the previous case, we have that $A= \Act$. Moreover for each $a \in \Act$, $L_r(n_a)\cdot \Act^{\omega}= \Act^{\omega}$. Following the same argument as above only for the $\no$ verdict we arrive at the conclusion that $\mathcal{E}_{\omega} \vdash n = \yes +\displaystyle\sum_{a \in \Act} a.\no = \yes +\no = m$. 
    \item The case  $m = \yes +\no \simeq_{\omega} \displaystyle\sum_{a \in A}a.n_a +\no$ is symmetrical to the one above. 
    \item Assume that $m = \yes + \displaystyle\sum_{a \in A} a.m_a \simeq_{\omega} \sum_{b \in B} b.n_b$ where both $m$ and $n$ are in reduced normal form. First of all, we follow an argument similar to the first case analyzed above, to the point where $\mathcal{E}_{\omega} \vdash n = \yes +\displaystyle\sum_{b \in B'} b.n_b'$ for some $\yes$-free monitors $n_b'$. 
For the proof of this final case we will use the following facts, whose validity can be easily established:  
\begin{enumerate}[start=1,label={(\bfseries S\arabic*)}]

        \item $B = \Act$,
        \item for all $b \in \Act$, $L_a(n_b) = \Act^{\omega}$, and
        \item for all $a \in A$, $L_r(m_a) = L_r(n_a)$. 
\end{enumerate}

So, for each $a \in A$, $\yes +m_a \simeq_{\omega} n_a$. 
Since both of these monitors have smaller depth that the original ones, we have that by induction: 
\begin{equation}\label{eq:1}
\mathcal{E}_{\omega} \vdash \yes+m_a = n_a,~ \forall~ a \in A \texttt{.} 
\end{equation}

For each $b \in \Act\setminus A$, we have that $\yes \simeq_{\omega} n_b$ (because $L_r^{\omega}(n_b) = \emptyset$). Again, we have that, by induction: 
\begin{equation}\label{eq:2}
\mathcal{E}_{\omega} \vdash \yes = n_a,~ \forall~ b \in \Act\setminus B \texttt{.}
\end{equation}

So: 
$$ \mathcal{E}_{\omega} \vdash n = \displaystyle\sum_{b \in \Act} b.n_b = \sum_{a \in A} a.n_a + \sum_{b \in \Act\setminus A} b.\yes $$
By equations (\ref{eq:1}) and (\ref{eq:2}): 
$$ \mathcal{E}_{\omega} \vdash n = \displaystyle\sum_{a \in A} a. (\yes + m_a) + \sum_{b \in \Act\setminus A} b.\yes$$

$$= \displaystyle\sum_{a \in A} a.\yes + \sum_{a \in A} a.m_a +\sum_{b \in \Act\setminus A} b.\yes  $$ 

$$ = \displaystyle\sum_{a \in \Act} a.\yes + \sum_{a \in A} a.m_a = \yes + \sum_{a \in A} a.m_a,$$ using axiom $Y_{\omega}$, and we are done. 

The above  analysis can be applied symmetrically for the cases: 
\begin{itemize}
    \item $m = \no + \displaystyle\sum_{a \in A} a.m_a \simeq_{\omega} \sum_{b \in B} b.n_b = n$ and
    \item $m = \displaystyle\sum_{a \in A} a.m_a \simeq_{\omega} \sum_{b \in B} b.n_b = n.$
\end{itemize}

% The only non-trivial case we will analyse is that one of theorem \ref{thm:comp}, case \ref{case:4}. 
% \textbf{Claim:} If $m =_{\omega}n$ then  $\forall s \in L_a(m)~ \exists k_0 : s.\displaystyle\sum_{a \in \Act} a.\sum_{a \in \Act}\ldots \yes, _{k_0 -times} \in L_a(n)$

% \begin{proof}
% Suppose that $\exists s_0 \in L_a(m): \forall k \geq 1 ~ s_0.\displaystyle\sum_{a \in \Act} a.\sum_{a \in \Act} a.\ldots a.\yes, _{k -times} \notin L_a(n) $ then by selecting a $k$ large enough so that the size of $s.s_k$ is larger than the depth of $n$ we know that $n$ cannot have a $\yes$ verdict for any sequence starting with $s.s_{k_0}$ while $m$ accepts all of them, which is a contradiction. 
\end{itemize}

This completes the proof. \end{proof}

% Therefore it suffices to apply axiom $Y_\omega$, $k_s$ times.

% The above procedure can be applied for any occurrence of $\yes$ and $no \in m,n$. Since each monitor will contain a finite amount of verdicts the above procedure will eventually produce two identical monitors. 
% \end{proof}
% \begin{note}
% One can confirm here that using axioms $Y_{\omega},N_{\omega}$ the families of axioms $Y_a,N_a$ are provable, while the converse does not hold. 
% \end{note}

\section{Open Terms}
% !TeX root = axioms_Monitors_VeEq.tex
%% Last modified: apr 21 10:07:34 GMT 2021

Thus far, we have only studied the completeness of equational axiom systems for $\simeq$ and $\simeq_{\omega}$ over closed terms. However, in our grammar we allow for variables and it is natural to wonder whether the ground-complete axiomatizations we have presented in Theorems \ref{thm:GrCompFin} and \ref{thm:GrCompOmega} are also complete for verdict equivalence and $\omega$-verdict equivalence over open terms. Unfortunately, this turns out to be false. Indeed, the equation 
\[
(\mathbf{O1})  ~~\yes +\no = \yes +\no + x
\]
is valid with respect to $\simeq$ (as both sides trivially accept and reject all traces), but cannot be proved using the equations in $\mathcal{E}_{\omega}$. This is because all the equations in that axiom system have the same variables on their left- and right-hand sides. Our goal in the remainder of this section is to study the equational theory of $\simeq$ and $\simeq_{\omega}$ over open terms. Subsection \ref{sect:openInf} will present our results when $\Act$ is infinite  as this case turns out to be more straightforward. We consider the setting of a finite set of actions in Subsection \ref{sect:openFin}. In what follows, we use $\mathcal{E}_v'$ for the axiom system that results by adding $O1$ to $\mathcal{E}_v$. The superscript $ '$ will be used in the name of an axiom set to denote that the axiom set is complete for one notion of equivalence over \textit{open} terms. The absence of a superscript refers respectively to a \textit{ground complete} axiom set. 

%The initial axiom system for open equations is $\mathcal{E}_v$ only expanded with $O1$ and will be called $\mathcal{E}_{v}'$.

Towards a completeness theorem, we modify the notion of normal form, to take variables into account. To that end we define: 

\begin{definition}\label{def:NormalFormOpen}
A term $m \in$ \MonF~is in \textbf{open normal form} if it has the form: 
$$m = \displaystyle\sum_{a \in A} a.m_a + \displaystyle\sum_{i \in I} x_i ~[+ \yes] ~ [+\no]$$ 
where $\{ x_i \mid i \in I\}$ is a finite set of variables, $A$ is a finite subset of $\Act$ and each $m_a$ is an (open) term in open normal form that is different from $\vend$.
\end{definition}

\begin{lemma}\label{lem:NormalFormOpen}
Each open term $m \in $\MonF~is provably equal to some open normal form $m'$ with $\depth(m') \leq \depth(m)$.

\end{lemma} 
The proof of the above result follows the lines of the one for Lemma \ref{lem:normClosed} for closed terms and is thus omitted.

%Note that each $m_a$ might contain by itself a set of variables where each one could appear in the set $\{ x_i \mid i \in I\}$ or not. 
%
%
%For instance consider the case: 

As in the case of closed terms, we now proceed to characterize a class of open normal forms for open terms whose verdict equivalence can be detected ``structurally''. The following example highlights the role that equation $(O1)$ plays in that characterization.

\begin{example}
Consider the following monitor in open normal form:
$$m = x+ \yes + a.b.(\no + b.a.x) .$$
Monitor $m$ contains two occurrences of the variable $x$. However, because of the interplay between the two verdicts, one of them is redundant and can be removed thus:

\begin{align*}
\mathcal{E}_v' \vdash m & =  x+ \yes + a.b.(\no + b.a.x) \\
& \overset{\mathrm{Y_a}}{=}  x+ \yes + a.\yes + a.b.(\no + b.a.x) \\
& \overset{\mathrm{Y_b}}{=}  x+ \yes + a.(\yes + b.\yes) + a.b.(\no + b.a.x) \\
&\overset{\mathrm{D_{a}}}{=}  x+ \yes + a.(\yes + b.\yes + b.(\no + b.a.x)) \\
&\overset{\mathrm{D_{b}}}{=}  x+ \yes + a.(\yes + b.(\yes + \no + b.a.x)) \\
& \overset{\mathrm{O_1}}{=}  x + \yes + a.(\yes + b.(\yes + \no)) \\
&\overset{\mathrm{D_{b}}}{=}  x + \yes + a.(\yes + b.\yes + b.\no) \\
& \overset{\mathrm{Y_b}}{=}   x + \yes + a.(\yes + b.\no) \\
&\overset{\mathrm{D_{a}}}{=}  x + \yes + a.\yes + a.b.\no \\
& \overset{\mathrm{Y_a}}{=}  x + \yes + a.b.\no .
\end{align*}

\end{example}

The above example motivates the following notion of reduced normal form for open terms.

\begin{definition} An \textbf{\textit{open reduced normal form}} is a term $$m = \displaystyle\sum_{a \in A} a. m_a + \sum_{i \in I} x_i~ [+\yes]~ [+\no]$$
where if $v \in \{ \yes,\no\}$ is a summand of m then each $m_a$ is $v$-free, different from $\vend$ and in open reduced normal form. In addition:

\begin{itemize}
\item if both $\yes$ and $\no$ are summands of $m$ then $m$ is equal to $\yes + \no$, 
\item if $\yes$ is a summand of $m$ and $m \xrightarrow[]{s} \no+m'$, for some $s$ and $m'$ then $m'$ is equal to $\vend$,
\item if $\no$ is a summand of $m$ and $m \xrightarrow[]{s} \yes+m'$, for some $s$ and $m'$ then $m'$ is equal to $\vend$.

\end{itemize}

\end{definition}

In what follows we will omit the word ``open'' when referring to the normal form of a term that contains variables. 
\begin{lemma}\label{lem:redNFOpen}
For each open monitor $m \in$ \MonF, its normal form  is provably equal to a reduced normal form.
\end{lemma}
\begin{proof}
By Lemma \ref{lem:NormalFormOpen} we may assume that $m$ is in in normal form. The proof is by induction on the size of $m$ and we isolate the following cases, depending on the verdicts $v \in \{\yes,\no\}$ $m$ has as summands:
\begin{enumerate}
	\item Case $m = \displaystyle\sum_{a \in A} a.m_a + \sum_{i \in I} x_i$. In this case we use the induction hypothesis on the $m_a$ monitors. These are different from $\vend$ and have smaller size than $m$ and therefore they are provably equal to a reduced normal form, i.e $\mathcal{E}_v' \vdash m_a = m_a'$ where $m_a'$ is in reduced normal form. Thus  $\mathcal{E}_v'$ proves $m = \displaystyle\sum_{a \in A} a.m'_a + \sum_{i \in I} x_i$, and we are done since $\displaystyle\sum_{a \in A} a.m'_a + \sum_{i \in I} x_i$ is in reduced normal form.

	By applying the congruence closure equational law we have that $\mathcal{E}_v' \vdash m = m'$, where $m'$ is in reduced normal form. 

	\item Case $m = \yes + \displaystyle\sum_{a \in A} a.m_a + \sum_{i \in I} x_i$. In this case by the induction hypothesis each $m_a$ is provably equal to a reduced normal form. The extra step here is that if $m \xrightarrow[]{s} \no+m'$, for some $s$ and $m'$ then $m'$ is equal to $\vend$. In such a scenario we have that:
	
If $s = \varepsilon$, then the claim follows trivially from $O1$. Otherwise $s=a.s'$ for some action $a \in \Act$ and $m \xrightarrow[]{a} m_a \xrightarrow[]{s'} \no+m'$. We now apply our axioms as follows: $$m = \yes + \displaystyle\sum_{a \in A} a.m_a + \sum_{i \in I} x_i \overset{\mathrm{Y_a}}{=} \yes + \displaystyle\sum_{b \in A \setminus\{a\}} b.m_b + a.\yes + a.m_a +\sum_{i \in I} x_i $$ $$\overset{\mathrm{D_a}}{=}\yes + \displaystyle\sum_{b \in A\setminus\{a\}} b.m_b + a.(\yes + a.m_a) +\sum_{i \in I} x_i.$$ 

This means that since $\yes +m_a$ has size smaller than $m$ it is provably equal to a reduced normal form. Additionally, since it contains a $\yes$ summand and $m_a \xrightarrow[]{s'} \no+m'$, by the induction hypothesis we have that $m'$ is equal to $\vend$ and we are done. 
% $$\yes +a.(x +\no +m_a') \overset{\mathrm{Y_a}}{=} \yes + a.\yes + a. (x +\no +m_a') \overset{\mathrm{D_a}}{=} \yes + a.(\yes +\no + x +m_a') $$ $$\overset{\mathrm{O_1}}{=} \yes + a. (\yes +\no) \overset{\mathrm{D_a}}{=} \yes +a.\yes +a.\no \overset{\mathrm{Y_a}}{=} \yes + a.\no$$\textbf{.}
 
    \item Case $m = \no + \displaystyle\sum_{a \in A} a.m_a + \sum_{i \in I} x_i$. The proof of this case is symmetrical to Case $2$ and therefore omitted. 
    \item Case $m = \yes + \no + \displaystyle\sum_{a \in A} a.m_a + \sum_{i \in I} x_i$. In this case we use the following simple argument. Starting for axiom $O_1$ we use the substitution $\sigma(x) = \displaystyle\sum_{a \in A} a.m_a + \sum_{i \in I} x_i$ and we get: $$\yes + \no  = \yes + \no + \displaystyle\sum_{a \in A} a.m_a + \sum_{i \in I} x_i,$$ and by applying the equational law of transitivity we have that $\mathcal{E}_v' \vdash m = \yes +\no$.
    \qedhere 
\end{enumerate}
\end{proof}
% Cases $2$ and $3$ are symmetric and slightly differ from the proof of Lemma \ref{lem:ReducedNF}. The extra claim in the case of the open terms is that we can remove occurrences of variables in the $m_a$'s when they contain a $\no$ summand (we present here the proof for case $2$).
% 
%
%Assume a monitor of the form $m = \yes + a.(x+\no + m_a')$. Then we apply our axioms as follows: $$\yes +a.(x +\no +m_a') = \yes + a.\yes + a. (x +\no +m_a') = \yes + a.(\yes +\no + x +m_a') $$ $$= \yes + a. (\yes +\no) = \yes +a.\yes +a.\no = \yes + a.\no$$\textbf{.}
% 
%The final case ($4$) is a simple application of our new axiom $O1$ to remove all the variables in $\displaystyle\sum_{i \in I} x_i$ and then repeating the inductive argument for the $m_a$'s.

The normal form defined above for open terms is adjusted over the closed terms case. This is because now our syntax is allowing for variables and therefore it is convenient for proofs to take these variables into account in a controlled and consistent manner. The further reducing that occurred towards defining the open reduced normal forms was possible due to the existence of the new axiom $O_1$, which gave us the option to remove variable occurrences. The new axiom $O_1$ is the only axiom we have currently available that does not contain every variable occurrence in both of its sides and it is therefore the only rule we have available that can help us remove variables from equations. In the presence of other axioms with this property we can further reduce our normal forms, as we will see later on.

%Note here that the defined normal form depends on the syntax we are analyzing in each case, while the reduced normal form depends on the axioms we have available for application (which in turn depend on our notion of equivalence and the characteristics of the model). Later on, where we will have extra axioms , our notion of reduced normal form will be further refined. 

In the following subsections, we will study the full equational theory of verdict and omega-verdict equivalence over open terms. 

%%%%%%%%%%%%%%%%%%%%%%%%%%%

\subsection{Infinite set of actions} \label{sect:openInf}
We begin by considering the equational theory of open monitors when
the set of actions is infinite. Apart from its theoretical interest,
this scenario has also some practical relevance. Indeed, as shown
already by Milner in~\cite{M80,Milner:1989:CC:63446}, infinite sets of
uninterpreted actions are useful when modeling system events that
carry data values. Runtime monitoring of systems with
data-dependent behavior has been an active field of research for over
15 years---see, for instance, the paper~\cite{BarringerGHS04} for an
early reference.

When the set of actions $\Act$ is infinite, it is easy to define a one-to-one mapping from open to closed terms that will help us prove completeness of the axiom system $\mathcal{E}_v'$.

\begin{theorem}\label{thm:CompOpenInf}(Completeness for open terms modulo $\simeq$) $\mathcal{E}_{v}'$ is complete for $\simeq$ over open monitors in \MonF~when $\Act$ is infinite. That is, for all $m,n \in$ \MonF, if $m \simeq n$, then $\mathcal{E}_{v}' \vdash m = n$.

\end{theorem}

\begin{proof}
Assume $m \simeq n$. By Lemma \ref{lem:redNFOpen}, we may assume that $m$ and $n$ are in reduced normal form. 

Let $$m = \displaystyle\sum_{a \in A} a.m_a + \sum_{i \in I} x_i ~[+\yes] ~[+\no] $$ and $$n =  \displaystyle\sum_{b \in B} b.n_b + \sum_{j \in J} y_j~ [+\yes] ~[+\no]\text{.}$$
We will show that $\mathcal{E}_v' \vdash m =n$ by induction on the sum of the sizes of $m$ and $n$. To this end, we will establish a strong structural correspondence between $m$ and $n$. Consider a substitution $\sigma$ defined as follows: $\sigma(x) = a_x.(\yes +\no)$ where \begin{itemize}
    \item for all variables $x$ and $y,~ a_x = a_y$ implies $x=y$, and 
    \item $\{a_x \mid x \in \Var \}$ is disjoint from the set of actions occurring in $m$ or $n$. 
\end{itemize}

Note that such a substitution $\sigma$ exists because $\Act$ is infinite. By induction on the sizes of $m$ and $n$, we will prove that if $\sigma(m) \simeq \sigma(n)$ then: 
\begin{enumerate}[start=1,label={(\bfseries C\arabic*)}]
    \item $v$ is a summand of $m$ iff $v$ is a summand of $n$, for $v \in \{\yes,\no\}$,
    \item $\{ x_i \mid i \in I  \} = \{ y_j \mid j \in J  \} $,
    \item $A=B$ and 
    \item for each $a \in A,~\sigma(m_a) \simeq \sigma(n_a)$.
\end{enumerate}

In what follows, we first show that $\mathcal{E}_v'$ proves $m = n$ assuming claims $\textbf{(C1)-(C4)}$ and then we prove those claims. To prove that $\mathcal{E}_v'$ proves $m = n$ follows from $\sigma(m) \simeq \sigma(n)$ for reduced normal forms $m$ and $n$, we proceed by induction on the sum of the sizes of $m$ and $n$. By claim $ \mathbf{C4}$, we have that $\sigma(m_a) \simeq \sigma(n_a)$ and, from the induction hypothesis, $\mathcal{E}_{v}' \vdash m_a = n_a$. By $\textbf{C1-3}$ we also have that $\mathcal{E}_{v}' \vdash \displaystyle\sum_{i \in I} x_i = \sum_{j \in J} y_j$ and that $\displaystyle\sum_{a \in A} a.m_a = \sum_{b \in B} b.n_b$, which means that by using the equational law of closure under summation we also have that $\mathcal{E}_{v}' \vdash m = n$.

%We will provide the proofs for the claims \textbf{$C1-4$} after we present our inductive argument bellow. By accepting that the above claims follow form $\sigma(m) \simeq \sigma(n)$, we can apply the following inductive argument: From \textbf{$C4$} we know that $\forall a \in A,~ m_a = n_a$ and from induction hypothesis $\mathcal{E}_{v}' \vdash m_a = n_a$. By $\textbf{C1-3}$ we also have that $\mathcal{E}_{v}' \vdash \displaystyle\sum_{i \in I} x_i = \sum_{j \in J} y_j$ ($\mathbf{C2}$) and that $\displaystyle\sum_{a \in A} a.m_a = \sum_{b \in B} b.n_b$ ($\mathbf{C3}$), which means that by using the equational law of closure under summation we also have that $\mathcal{E}_{v}' \vdash m = n$.

We present now the proofs of $\textbf{(C1)-(C4)}$.

\textbf{C1:} Assume $\yes$ is a summand of $m$. Then $\varepsilon \in L_a(\sigma(m))$. Since $\sigma(m) \simeq \sigma(n)$, we have that  $\varepsilon \in L_a(\sigma(n))$. Note that $\varepsilon \not\in L_a(\sigma(x))$ for each $x$. Thus $\yes$ must be a summand of $n$. The case for $v = \no$ is similar. By symmetry the claim follows. 
%
%From the normalization lemma for open terms and Lemma \ref{lem:ReducedNF}, it must be the case that if the verdict $v \in \{\yes,\no\}$ is a summand of $m$ then $m$ will be in one of the following forms:
%\begin{itemize}
%    \item $m = \yes +\no$. 
%    
%    \item $ m = \yes + \displaystyle\sum _{a \in A} a.m_a + \sum_{i \in I} x_i$, where each $m_a$ is $``\yes$''-free and in reduced normal form.
%    \item $ m = \no + \displaystyle\sum _{a \in A} a.m_a + \sum_{i \in I} x_i$, where each $m_a$ is $``\no$''-free and in reduced normal form.
%\end{itemize}
%In the first case we have that $n = \yes +\no$ as no other reduced normal form is equal to $\yes +\no$.         
%
%
%The two following cases are symmetrical so we only discuss the first. We have that $n$ must be equal to a form $n = \yes + \displaystyle\sum _{b \in B} b.n_b + \sum_{j \in J} y_j$ (where the $n_b$'s are $``\yes$''-free) as both monitors need to accept the empty trace. 
%

\textbf{C2:} Assume that  $x \in \{x_i \mid i \in I\}$.  By the definition of $\sigma$, it follows that $\sigma(m)$ both accepts and rejects the trace $a_{x}$. Since $m \simeq n$, we have that $\sigma(n)$ also accepts and rejects the trace $a_{x}$. As $n$ does not contain any occurrence of $a_{x}$ and has  at most one of the verdicts $\yes$ and $\no$ as a summand,  it follows that $x \in \{ y_j \mid j \in J  \}$. Therefore, by symmetry, $ \{x_i \mid i \in I\} = \{ y_j \mid j \in J  \}$ and we are done. 

\textbf{C3:}
Assume, towards a contradiction, that $a \in A \setminus B$. Then $m$ cannot have both $\yes$ and $\no$ as summands, since $m$ is in reduced normal form.

If $m$ has none of the verdicts as a summand, we know that $m_a$ is
different from $\vend$ since $m$ is in reduced normal form. Therefore
$\sigma(m_a)$ will either accept or reject some trace $s$, which
implies that $\sigma(m) $ will also accept or reject $as$. However,
$\sigma(n)$ cannot do the same because $a \not\in B$,
$\sigma(x) \not\xRightarrow[]{a}$ for each $x$, and neither $\yes$ nor
$\no$ are summands of $n$. This contradicts our assumption that
$m\simeq n$.

Assume now, without loss of generality, that $m$ has only the verdict
$\yes$ as summand. Observe that $m_a$ is
$\yes$-free and different from $\vend$, since $m$ is in reduced normal
form. This means that $\sigma(m_a)$ can reject some trace $s$ and,
therefore, that $\sigma(m) $ will reject $a.s$. On the other hand,
$\sigma(n)$ cannot do the same because $a \not\in B$,
$\sigma(x) \not\xRightarrow[]{a}$ for each $x$ and $\no$ is not a
summand of $n$. Again, this contradicts our assumption that $m\simeq
n$.

The above analysis yields that $A \subseteq B$. By symmetry, $A=B$ follows.

\textbf{C4:} Our final claim (and the one with the most involved proof) is that  $\sigma(m_a) \simeq \sigma(n_a)$, for each $a \in A$.
%%Since $m$ and $n$ are in reduced normal form we have that if they
%%contain both non-$\vend$ verdicts as summands then they are both provably
%%equal to $ \yes + \no$ and we are done. 

 If the reduced normal forms of the monitors do not contain any verdict $v \in \{\yes,\no \}$ as a summand, then the argument is simplified significantly. Therefore, we limit ourselves to presenting here the most complicated case, where $m$ and $n$  both contain exactly one verdict $v \in \{yes,\no \}$ as a summand.  Without loss of generality, we assume that this verdict is $\yes$, namely that
 $$m = \yes + \displaystyle\sum_{a \in A} a.m_a + \sum_{ i \in I} x_i $$ and $$n = \yes + \displaystyle\sum_{b \in B} b.n_b + \sum_{ j \in J} y_i .$$

Since the claims $\textbf{C1-3}$ have already been proven, we know for $m$ and $n$ that:  

$$ m = \yes + \displaystyle\sum_{a \in A} a.m_a + \sum_{ i \in I} x_i \text{ and } n = \yes + \displaystyle\sum_{a \in A} a.n_a + \sum_{ i \in I} x_i ~.$$

We remind the reader that our purpose is to prove that $\sigma(m_a) \simeq \sigma(n_a)$, for each $a \in A$, so that we can apply our induction hypothesis to infer that $\mathcal{E}_v' \vdash m_a = n_a$.

We first prove that the rejection sets of $\sigma(m_a)$ and $\sigma(n_a)$ are equal. To this end, assume that  $s\in L_r(\sigma(m_a))$. It follows that $a.s \in L_r(\sigma(m)) = L_r(\sigma(n))$. By the form of $n$ and from the definition of $\sigma$, we conclude that $s \in L_r(\sigma(n_a))$. Therefore,  $L_r(\sigma(m_a)) \subseteq L_r(\sigma(n_a))$. By symmetry we have that $L_r(\sigma(m_a)) = L_r(\sigma(n_a))$ and we are done.

It remains to prove that the acceptance sets of $\sigma(m_a)$ and $\sigma(n_a)$ are also identical. 
%%%%%%%%%%%%%%%%%%%%%%
% By induction hypothesis $\mathcal{E}_v' \vdash m_a \simeq n_a$ for all $a \in A$  and therefore we have claim $\textbf{C4}$.
%%%%%%%%%%%%%%%%
%%% To that end, we present the following argument.
(It is important here to point out that, since both $m$ and $n$
contain a $\yes$ verdict as a summand, the acceptance sets of
$\sigma(m)$ and $\sigma(n)$ are both equal to $\Act^*$. However, for
our inductive argument to work, we need to be able to prove that
$L_a(\sigma(m_a)) = L_a(\sigma(n_a))$.) To that end and towards a contradiction, consider a shortest trace
$s$ that is accepted by monitor $\sigma(m_a)$, but not by $\sigma(n_a)$. Consequently, monitor
$\sigma(m)$ accepts the trace $a.s$.
% Due both the $\yes$ verdict at the top level but also though the $\sigma(m_a)$. 

Since monitors $m_a$ and $n_a$ are $\yes$-free, as a result of $m$ and $n$ being in reduced normal form, the acceptance of $s$ must be the result of a variable $x$ mapped to $a_x.(\yes +\no)$ through the substitution $\sigma$.
%%%%%%%
%In order to simplify the proof we assume that $s$ is the shorer trace we
%can find of this property, i.e. that it is accepted by $\sigma(m_a)$ but
%not $\sigma(n_a)$.
%%%%%%%%%%%%%%
Since $s$ is a shortest trace that is accepted by monitor $\sigma(m_a)$, but not by $\sigma(n_a)$, none of its prefixes is accepted by $\sigma(m_a)$ and therefore the last action that is in $s$ must be the action $a_x$ stemming from $\sigma(x)$. This means that monitor $m_a$ can perform the transition $m_a \xRightarrow[]{s'} m_a'$,  where $m_a'$ contains $x$ as a summand and $s=s'.a_x$. Therefore the monitor $\sigma(m_a)$ can perform the transitions: 
$$ \sigma(m_a) \xRightarrow[]{s'} \sigma(m_a') \xrightarrow[]{a_x} \yes + \no \xrightarrow[]{\tau} \yes .$$ 

%%We argue that $\sigma(n_a)$ can also accept $s$.
%%%%%%%%%%%
Since $s'.a_x$ is accepted by $\sigma(m_a)$, it must also be rejected by it because $a_x$ is an action that can only be observed after the substitution of the variable $x$ in $m_a$. We have already argued that the rejection sets of $\sigma(m_a)$ and $\sigma(n_a)$ are equal and therefore $\sigma(n_a)$ also rejects the trace $s'.a_x$. Since the action $a_x$ is a unique action corresponding to the variable $x$, there are only two ways in which $\sigma(n_a)$ could reject the trace $s'.a_x$. The first case is that $\sigma(n_a)$ can also perform the transitions

$$\sigma(n_a) \xRightarrow[]{s'} \sigma(n_a') \xrightarrow[]{a_x} \yes + \no$$ wfor some $n_a'$. However, this would guarantee that $\sigma(n_a)$ accepts $s$, whereas we assumed that it does not. 

The most complicated case is when $\sigma(n_a)$ can reject a prefix $s_0$ of $s'$. By the already proven equality of the rejection sets of the two sub-monitors, $\sigma(m_a)$ would also reject $s_0$. This can only happen if both $n_a$ and $m_a$ rejected that prefix independently of the substitution $\sigma$, since every action preceding $a_x$ along the trace $s'$ is not an action corresponding to the mapping of a variable through $\sigma$ as explained above. This means that both $m_a$ and $n_a$ can perform the transitions $m_a \xRightarrow[]{s_0} \no +m_a'$ and $n_a \xRightarrow[]{s_0} \no +n_a'$, for some $m_a'$ and $n_a'$. However, since $m$ and $n$ are in reduced normal form, this implies that $m_a'$ and $n_a'$ are equal to $\vend$. This leads us to a contradiction, as we assumed that $\sigma(m_a)$ accepts the trace $s$ which can no longer be the case if $m_a \xrightarrow[]{s_0} \no + \vend$ where $s_0$ is a prefix of $s$.

Therefore every trace accepted by $\sigma(m_a)$ is also accepted by $\sigma(n_a)$. By symmetry, we have that the acceptance sets of $m_a$ and $n_a$ are equal.

This means that $\sigma(m_a) \simeq \sigma(n_a)$, which completes the proof of \textbf{C4} and consequently of the whole theorem. \end{proof}
 
%The final case of the proof where $m.n$ do not accept or reject the empty trace (that is, when they don't have $\yes$ or $\no$ summands) follows the same structure for the completeness proof of closed terms by using the modification presented above. Namely by the use of axiom $D_a$ we construct two summands of the monitors, one $\yes$-free and one $\no$- free. 

\begin{corollary}
$\mathcal{E}_{v}'$ is complete for $\simeq_{\omega}$ over open monitors in \MonF~when $\Act$ is infinite. That is, for all $m,n \in$ \MonF, if $m \simeq n$, then $\mathcal{E}_{v}' \vdash m = n$.
\end{corollary}

\begin{proof}
The claim follows from Lemma \ref{lem:1}. \end{proof}

\subsection{Finite set of actions}\label{sect:openFin}
% !TeX root = axioms_Monitors_VeEq.tex
%% Last modified: apr 19 10:50:31 GMT 2021
%% Last spell checked: 

The study of the equational theory of $\simeq$ when $\Act$ is finite turns out to be more interesting and complicated. In this setting, we can identify equations whose validity depends on the cardinality of $\Act$, which is not the case for any of the axioms we used so far. To see this, consider the equation 
\[
(\mathbf{V_{1}}) ~~ x = x +a.x ,
\]
which is sound when $\Act = \{ a\}$ but cannot be derived by the equations in $\mathcal{E}_{v}'$, as it is not sound when $\Act \neq \{ a\}$. 
%Such equations were not a problem when the action set was infinite because we had a somewhat broader space for the substitutions we could define and therefore "less" equations were sound.

As a first step in our study of the equational theory of $\simeq$ when $\Act$ is finite, we characterize some properties of sound equations.

\begin{lemma}\label{lem:firstvar}

Let $m\simeq n$ be a sound equation, where  $m,n \in$ \MonF       ~and $m$ is in reduced normal form. Assume that
\begin{itemize}
\item $m \xrightarrow[\text{}]{s} x + m'$,  for some $s$ in $\Act^*$, variable $x$ and $m'$ in \MonF,
and
\item $m \not \xrightarrow[\text{}]{s_p} x+ m_{s_p}$, for each proper prefix $s_p$ of $s$ and $m_{s_p} \in$ \MonF.
\end{itemize}
Then, $n \xrightarrow[\text{}]{s} x + n'$ for some $n'$ in \MonF.

\end{lemma}

\begin{proof}
Consider the substitution 
$$ \sigma(y) = 
\begin{cases}
\yes + \no, & \text{if} ~y = x \\
  \vend,  & \text{if} ~ y \neq x.
\end{cases} $$

Since $m \xrightarrow[\text{}]{s} x + m'$ by one of the assumptions of the lemma, we have that $\sigma(m)$ will both accept and reject $s$. Since $m \simeq n$ is sound we have that $\sigma(n)$ must do the same. If $n \not\xrightarrow[\text{}]{s} x + n'$ for every $n'$ then it is not hard to see that there are two ways in which $n$ could accept and reject $s$: 
\begin{enumerate}
\item $n \xRightarrow[\text{}]{s'} \yes$ and $n \xRightarrow[\text{}]{s'} \no$ where $s'$ is a prefix of $s$ (including $s$ itself), or
\item $n \xrightarrow[\text{}]{s'} x + n'$ where $s'$ is a prefix of $s$ (so that $\sigma(n)$ would accept and reject $s'$ and therefore $s$).
\end{enumerate}

In the first case, consider the substitution $\sigma_e$ that maps all variables to $\vend$. Since $n \xRightarrow[\text{}]{s'} \yes$ and $n \xRightarrow[\text{}]{s'} \no$, we have that $\sigma_e(n)$  accepts and rejects $s'$. From $m \simeq n$, we have that  $\sigma_e(m)$ also accepts and rejects $s'$. It is not hard to see that this means that $m \xRightarrow[\text{}]{s'} \yes$ and $m \xRightarrow[\text{}]{s'} \no$ However, this is impossible because $m$ is a reduced normal form and $m \xrightarrow[\text{}]{s} x + m'$ by the proviso of the lemma.

%In the first case we have that $n \xRightarrow[\text{}]{s'} \yes$ and $n \xRightarrow[\text{}]{s'} no$ which means that $n$ would accept and reject the trace $s$ even under the substitution $\sigma(x) = \vend~, \forall x$. By verdict equivalence of $m$ and $n$ the same should apply for $m$, i.e. $m \xRightarrow[\text{}]{s'} \yes$ and $m \xRightarrow[\text{}]{s'} \no$. However this is not allowed since $m \xRightarrow[\text{}]{s} x$ and it is in reduced normal form. Therefore it cannot be the case where  $n \xRightarrow[\text{}]{s'} \yes$ and $n \xRightarrow[\text{}]{s'} \yes$for any prefix $s'$ of $s$ including $s$ itself. 

In the second case, even though both monitors accept and reject $s$, we also have that $\sigma_e(n)$ also accepts and rejects $s'$. Again, since the two monitors are verdict equivalent, we know that $\sigma_e(m)$ must do the same. Since $m$ is in reduced normal form and $m \not \xRightarrow[\text{}]{s_p} x + m'$ for any prefix $s_p$ of $s$ (and therefore neither for $s'$) we have that $\sigma(m)$ can only accept and reject $s'$ by performing the transitions $ m \xRightarrow[\text{}]{s_1''} \yes$ and $ m \xRightarrow[\text{}]{s_2''} \no$, for $s_1''$ and $s_2''$ prefixes of $s'$. This however is not allowed since it contradicts the fact that $m$ is in reduced normal form and $m \xrightarrow[\text{}]{s} x + m'$.
Since both cases have led to a contradiction, we can infer that there is some $n'$ such that $n \xRightarrow[\text{}]{s} x + n'$, which was to be shown. 
 \end{proof}

\begin{corollary}\label{cor:firstVarConvinient}

Let $m\simeq n$ be a sound equation, where  $m,n \in$ \MonF~and $m$ is in reduced normal form. Assume that
\begin{itemize}
\item $m \xrightarrow[\text{}]{s} x + m'$,  for some $s$ in $\Act^*$, variable $x$ and $m'$ in \MonF,
and
\item $n \not \xrightarrow[\text{}]{s} x+ n'$, for any $n' \in$ \MonF.
\end{itemize} then we have that there exists an $s_p$ prefix of $s$ such that
\begin{itemize}
\item $m  \xrightarrow[\text{}]{s_p} x+ m_{s_p}$ and $n  \xrightarrow[\text{}]{s_p} x+ n_{s_p}$ for some $m_{s_p}$ and  $n_{s_p}$ in \MonF, and 
\item for any prefix $s_0$ of $s_p$ we have that $m \not\xRightarrow[\text{}]{s_0} x + m'$ and $n \not\xRightarrow[\text{}]{s_0} x + n'$ for any $m'$ and $n'$ in \MonF.
\end{itemize} 
\end{corollary}

\begin{proof}

Assume a sound equation $m \simeq n$ for witch we have $m \xrightarrow[\text{}]{s} x + m'$,  for some $s$ in $\Act^*$, variable $x$ and $m'$ in \MonF. If this is the first occurrence of $x$ along the trace $s$ in $m$ (i.e. $m \not \xrightarrow[\text{}]{s_p} x+ m_{s_p}$, for each proper prefix $s_p$ of $s$ and $m_{s_p} \in$ \MonF), then by Lemma \ref{lem:firstvar}, we would have that $n$ must be able to perform the transitions $n \xrightarrow[\text{}]{s} x+ n'$, for some $n'$ in \MonF. Since this cannot be the case as the proviso of the corollary forbids it we have that there must be a prefix $s_p$ of $s$ such that $m  \xrightarrow[\text{}]{s_p} x+ m_{s_p}$.

Without loss of generality we assume $s_p$ to be the shortest such trace, which means there are no other occurrences of the variable $x$ along the trace $s_p$. We can therefore see that now for the trace $s_p$, Lemma \ref{lem:firstvar} holds and therefore $n  \xrightarrow[\text{}]{s_p} x+ n_{s_p}$ for some $n_{s_p}$ in \MonF. Additionally since we assumed $s_p$ t be the shortest trace of the necessary property we already have that $m \not\xRightarrow[\text{}]{s_0} x + m'$ for any $m'$ in \MonF. 

It remains to show that the same must hold for $n$. This can be easily seen to be the case since if we assumed the opposite where for some prefix $s_0$ of $s_p$ we had $n  \xrightarrow[\text{}]{s_0} x+ n_0$ for some $n_0$ then by the symmetric analysis and by using the previous lemma and this corollary we would arrive at a contradiction of $s_p$ being the shortest prefix of $s$ for witch $m  \xrightarrow[\text{}]{s_p} x+ m_{s_p}$. 
\end{proof}
\begin{remark}
In what follows, when studying open equations, we will refer to occurrences of variables such as the one mentioned in the above corollary, where only one of the monitors involved in the equation can reach a term of the form $x + m_x$ after observing a trace $s$, as ``one-sided'' variable occurrences. 
\end{remark}
 
Intuitively Lemma \ref{lem:firstvar} states that on each sound equation (including axioms) of which at least one side is in reduced normal form, the first occurrence of each variable per distinct trace leading to the variable is common for both sides of the equation.This gives us some handy intuition on what restrictions an equation that is sound must satisfy.

The following example shows Lemma \ref{lem:firstvar} in action.  

\begin{example}\label{ex:onesidedvague}
The equation \[ 
x + a.(x + a.(\yes +\no) + b.(\yes +\no)) = x + a.(a.(\yes +\no) + b.(\yes +\no))~\]

 is sound over the set of actions $\Act = \{ a,b \}$, but $$x + a.(x + a.(\yes +\no) + b.(\yes +\no)) = a.(x + a.(\yes +\no) + b.(\yes +\no))$$ is not since the first occurrence of the variable $x$ in the second example happens after the prefix $\varepsilon$ on the left-hand side but after the prefix $a$ on the right. In the second equation, the earliest occurrence of the variable $x$ (after the prefix $\varepsilon$) is one-sided.

Also notice here the importance of the sub-term $a. \displaystyle\sum_{a \in \Act}a.(\yes +\no)$. We will see that this type of sub-term is crucial for the soundness of the open equations with one-sided variable occurrences  we encounter later on. 
\end{example} 
The following notation will be used in what follows to describe a family of sound equations that generalize the one given in Example \ref{ex:onesidedvague}.
\begin{definition} (\textbf{Notation}) Let $s \in \Act^*$.

\begin{enumerate}
\item We use $\pre(s)$ to denote the set of prefixes of $s$ (including $s$). 
\item We use $s^i$, $i \geq 1$, to denote the trace $s$ if $i = 1$ and $ss^{i-1}$ otherwise.  
\item We use $s.m$ to stand for a monitor that can perform exactly the actions along the finite trace $s$ and then become $m$.
\item We define $
  \overline{s}^{\leq}(m) = \displaystyle\sum_{\substack{|s'|\leq |s|,~ \\ s'\not\in \pre(s) }} s'.m  
$
.
The monitor $\overline{s}^{\leq}(m)$ is one that behaves like $m$ after having observed any trace of length at most $|s|$  that is not a prefix of $s$. 

\item The term $\overline{s}(m)$ is defined thus: $ \overline{s}^{\leq}(m) + s.\displaystyle\sum_{a \in \Act } a.m$~.

Intuitively $\overline{s}(\yes + \no)$ stands for the monitor that accepts and rejects all traces that \emph{do not} cause the acceptance or rejection of the string $s$. Those are exactly the traces that are shorter than $s$ but not its prefixes, and also the ones extending $s$.

\item With the term $\overline{s}^{(k)}(m)$, for $k \geq 1$, we will mean the summation: 

 $ \overline{s}(m)~ \text{if} ~k = 1~\text{and}~\displaystyle\sum_{1 \leq i < k-1} s^i.\overline{s}^{\leq}(m) + s^{k-1}.\overline{s}(m)$ if $k \geq 2$.

Intuitively $\overline{s}^{(k)}(\yes + \no)$ stands for a monitor that, after observing the fixed trace $s$,  accepts and rejects everything except the trace $s^k$ (and its prefixes).  
%\item $T(m)$ where $T$ is a set of traces stnds for the summation $\displaystyle\sum_{s \in T} s.m$. 
%\item $\left.{T}\right|_A(m)$  where $T$ and $A$ are sets of traces, stands for $\displaystyle\sum_{s \in T\setminus A} s.m$.
\end{enumerate}
\end{definition}

We now present an example of the usage of the above notation in order to help the reader understand the equations presented later involving these new notions.

\begin{example}\label{ex:onesidedexample}
For a set of actions $\Act = \{ a,b\}$, the monitor $m = \yes + \no$ and a trace $s = ab$ we have that: 

\begin{itemize}
\item $\mathit{pre}(s)  = \{\varepsilon,a,ab \}$

\item $ \overline{s}^{\leq}(m) = b.(\yes+\no) + a.a.(\yes +\no) + b.b.(\yes +\no) + b.a.(\yes+\no) $

\item $\overline{s}(m) = b.(\yes+\no) + a.a.(\yes +\no) + b.b.(\yes +\no) + b.a.(\yes+\no) + a.b.\displaystyle\sum_{c \in \Act} c.(\yes + \no)$

\item and for $k = 3$ we get $$\overline{s}^{(3)}(m) = s.\overline{s}^{\leq}(m) + s^2.\overline{s}(m) = $$ $$a.b.(b.(\yes+\no) + a.a.(\yes +\no)) +$$ $$ a.b.a.b.(b.(\yes+\no) + a.a.(\yes +\no) + b.b.(\yes +\no) +$$ $$  b.a.(\yes+\no)  + a.b.\displaystyle\sum_{c \in \Act} c.(\yes + \no))$$
\end{itemize}
\end{example}
This notation defined and presented above is very useful once one understands a very particular form equations among open monitors take when they involve one-sided variable occurrences. Consider, for instance,  the following sound equation (for a fixed constant $k$):

$$x + a^k.x + \overline{a^k}^{3}(\yes+\no) \simeq x + \overline{a^k}^{3}(\yes+\no) ~.$$

We will formally prove the soundness of (a more general form of) this equation later on. We can intuitively see from the examples above that when an equation contains a one-sided variable occurrence, then the rest of the terms involved in the equation must have some specific form as well so that the equation will stay sound under all possible substitutions. This means that certain traces must always be accepted and rejected by both sides independently of a substitution.

The following lemma formalizes this intuition.
\begin{lemma}\label{lem:sumAfterVar}
Assume $m \simeq n$, where $m,n$ are in reduced normal form. If $m \xrightarrow[\text{}]{s} x + m'$ for some $m'$ but $n \not\xrightarrow[\text{}]{s} x+ n'$ for any $n'$, then there exist $s',s''$ such that $s = s's''$ and, for all $s_b = ss_p$ where $s_p \not\in \pre(s'')$, either: 
\begin{itemize}
\item $m \xRightarrow[\text{}]{s_b} \yes$, $m \xRightarrow[\text{}]{s_b} \no$, $n \xRightarrow[\text{}]{s_b} \yes $ and $n \xRightarrow[\text{}]{s_b} \no$ or
\item $\exists s_0, m'', n''$ such that  $m \xrightarrow[\text{}]{s_0} x + m'' $ and $n \xrightarrow[\text{}]{s_0} x + n''$ and $s_0.s_b \in \pre(s.s_b)$.
\end{itemize}
\end{lemma}

%For any sound equation $e:~ m \simeq n$ if $s.x \in m$ and $s.x \not\in n$ then there exists a prefix of $s^{i}$ for some $i$, $s'$ ($s^{i} = s'.s''$),such that for all $s_b \in s^{i}.\overline{s'}$, either $s_b.(\yes+\no) \in m,n$ or $\exists s_0'$ s.t. $s_0'.x \in m,n$ and $s_0'.s_b$ is a prefix of $s.s_b$. 

\begin{proof}
We have an equation $m \simeq n$, with $m$ and $n$ in reduced normal form, for which we assume that: $m \xrightarrow[\text{}]{s} x + m'$ but $n \not\xrightarrow[\text{}]{s} x+ n'$ for any $n'$. Let $s$ be the shortest trace meeting the proviso of the lemma. It is not hard to see that $s \neq \varepsilon$ because $m \simeq n$ and $m$ and $n$ are in reduced normal form. This means that indeed in the monitors $m,n$ all other earlier occurrences of $x$ happen at both sides. 
By Corollary \ref{cor:firstVarConvinient} we know that there is a prefix of $s$ called $s'$ ($s = s'.s''$) such that both $m$ and $n$ can perform the transitions $m \xrightarrow[\text{}]{s'} x + m_0'$ and $n \xrightarrow[\text{}]{s'} x + n_0'$, and in addition for every prefix of $s'$ we have that  $n \not\xrightarrow[\text{}]{s'} x+ n'$ and $m \not\xrightarrow[\text{}]{s'} x+ m'$ for every $m'$ and $n'$.

This means that there are no other one-sided occurrences of the variable $x$ ``between'' $s'$ and $s''$ (otherwise $s$ would not be the shortest trace). Since $m \simeq n$ is sound, we know that under any substitution the resulting monitors are verdict equivalent. 

Consider the set of traces  $$A = \{ t  \mid (|t| \leq |s''|\wedge t \not\in \pre(s''))\vee t = s''.t',~ t' \in \Act^+ \}~. $$

We now associate with this set of traces the class $\mathcal{S}_A$ of substitutions $\sigma$ as the ones that for at least one trace $s_p \in A $ we have that $\sigma(x) \xRightarrow[\text{}]{s_p} \yes$ or $\sigma(x) \xRightarrow[\text{}]{s_p} \no$. Note that the class of substitution $\mathcal{S}_A$ contains many substitution for each trace $s_p$ and, additionally, since the set $A$ is infinite, $\mathcal{S}_A$ is infinite as well. 

Fix now a $s_p$ and a substitution $\sigma \in \mathcal{S}_A$ such that $\sigma(x) \xRightarrow[\text{}]{s_p} \yes$. We have therefore that $\sigma(m) \xRightarrow[]{s's_p} \yes$, $\sigma(n) \xRightarrow[]{s's_p} \yes$  and $\sigma(m) \xRightarrow[\text{}]{ss_p} \yes$. By the construction of $A$, $s's_p$ is not a prefix of $ss_p$ and therefore it is not necessary that $\sigma(n) \xRightarrow[\text{}]{ss_p} \yes$. However, since $m \simeq n$ is sound we have that $\sigma(n)$ must also be able to accept $s.s_p$. One way this could happen is if both monitors, $m$ and $n$ accept and reject the trace $s_b = s.s_p$ where $s_p \in A$ independently of a substitution, i.e. $m \xRightarrow[\text{}]{s_b} \yes$, $m \xRightarrow[\text{}]{s_b} \no$, $n \xRightarrow[\text{}]{s_b} \yes $ and $n \xRightarrow[\text{}]{s_b} \no$. Note that if one monitor can perform these transitions independently of a substitution then the other one must do so as well since they are verdict equivalent. If this is the case then for the traces $s',s''$ with $s = s's''$ an for all $s_b = ss_p$ where $s_p \not\in \pre(s'')$ the first bullet of the lemma holds.

If this is not the case however we have that for a trace $s_p$ and a substitution $\sigma \in \mathcal{S}_A$ such that $\sigma(x) \xRightarrow[\text{}]{s_p} \yes$ the monitor $n$ must somehow accept the trace $ss_p$ and this is not done because $n \xRightarrow[\text{}]{s_b} \yes$.

We remind to the reader here that $s$ is the shortest we could find that satisfied the proviso of the lemma. Therefore there are no other one-sided variable occurrences along the trace $s$. 

This means that the only way than $n$ could accept $s_b$ is another variable occurrence (not one-sided as $s$ is the shortest trace satisfying the proviso of the the lemma) happening after some other prefix $s_0$ of $s$. I.e. $n \xrightarrow[\text{}]{s_0} x + n_1$,  $m \xrightarrow[\text{}]{s_0} x + m_1$ for some monitors $n_1$ and $m_1$ and trace $s_0.s_p$ is a prefix of $s'.s''.s_p  = s.s_p$. Note here that by Corollary \ref{cor:firstVarConvinient} we know that $s'$ is the shortest trace after which the variable $x$ occurs. Therefore our only options for the trace $s_0$ would be the trace $s'$ and its extensions which falls in the second case of the lemma as $s_0.s_p$ is a prefix of $s.s_p$.  

This concludes the case analysis for the shortest $s$ leading to a one-sided variable occurrence of a variable. We continue with a trace $s_1$ as the immediately longer than $s$. For this $s_1$ with $|s_1| \geq |s|$ we can generalize the result as follows: 

If $s \in \pre(s_1)$ then the trace $s'$ we identified with the case analysis $s$ is also a prefix of $s_1$ (i.e. $s_1 = s'.s_1''$) and the same transitions we proved for the traces $s_b$ are also enough for the result to hold for the trace $s_1$.  Assume now that $s \not\in \pre(s_1)$. Then Corollary \ref{cor:firstVarConvinient} still holds and the one-sided variable occurrence after the trace $s_1$ also does not have any other one-sided variable occurrences between itself and the prefix guaranteed by the corollary which means we can apply the same analysis. \end{proof} 

\subsubsection{Completeness of verdict equivalence}\label{sect:ComplOpenFin}

%\paragraph{Verdict equivalence for open terms}

In this section we will present our axiom system for open monitors over a finite number of actions. We start by providing an axiom set, which we prove to be sound and complete for verdict equivalence over \MonF. In order to do so, we first use these axioms to further reduce a normal form of a term. Then, by utilizing this new reduced normal form we use structural induction to prove the completeness of our axiom set. The axiom set we provide is infinite. It is therefore natural to ask whether $\simeq$ is finitely axiomatizable over \MonF. We answer this question negatively by proving that no complete finite axiom set exists for this algebra. This final part follows a different type of argument which we will present in Section \ref{sec:neg}.

% We first identify some useful quantities of any equation between monitors and we use the infinite family of equations provided earlier to show that there are sound equations of arbitrarily large such quantities. We then use induction on the length of the proof to show that any sound equation derived from a finite set of axioms will have an upper bound on these quantities associated with it. Since we know that valid equations of arbitrarily large such quantities exist we conclude that no finite family of axioms can be complete for verdict equivalence in \MonF.  

When studying open equations over a finite set of actions one would hope that one of the axiom systems presented already would be complete. However, we can guarantee that the equations provided in $\mathcal{E}_v'$ are definitely unable to prove every sound open equation. To see this consider the equation used in Example \ref{ex:onesidedexample} (where $k$ is a constant): $$x + a^k.x + \overline{a^k}^{(3)}(\yes+\no) \simeq x + \overline{a^k}^{(3)}(\yes+\no) ~.$$

We can clearly see that one of the sides of this equations contains a one-sided variable occurrence (remember that we are considering terms up to $A1-A4$). The only axiom which has a similar behavior is $O1$. However for axiom $O1$ to be applied it must be the case that a variable is occurring simultaneously with a $\yes$ \textit{{and}} a $\no$ verdict. Since this does not apply for the equation we are examining it is easy to see that no proof involving only the axioms of $\mathcal{E}_v'$ could prove it. 
 
Towards proving this kind of equations and when $\Act$ is finite, we consider the family of axioms

 %With the term $\overline{s}$ we will denote the set of traces $\mathcal{S} = \{w \mid w~is~not~a~prefix~of~s \} $

\label{AxiomsOpen}

 $$\mathcal{O} = \{ O2_{s,k} \mid s \in \Act^*, k \geq 0 \}$$ where 

$${(\mathbf{O2_{s,k}})} ~~ x + s.x + \overline{s}^{(k)}(\yes + \no)  = x + \overline{s}^{(k)}(\yes + \no) ~.$$

We extend our finite axiom set $\mathcal{E}_v'$ for open terms to the infinite $\mathcal{E}_v' \cup \mathcal{O}$, which we will call $\mathcal{E}_{\textit{v,f}}'$. The subscript $f$ in the naming scheme states that the action set for which the axiom system is complete is finite. When the action set is a singleton, we will replace it with the subscript $1$. If the cardinality of the action set is not important, or if it is infinite, then we use no subscript. Based on the naming scheme we have defined, the name of the axiom set $\mathcal{E}_{\textit{v,f}}'$ denotes that we are studying verdict equivalence ($_v$), over open terms ($'$) and for a finite set of actions ($_f$).

\begin{lemma}
$\mathcal{E}_{\textit{v,f}}'$ is sound. That is, if $\mathcal{E}_{\textit{v,f}}' \vdash m = n$ then $m \simeq n$, for all $m,n \in $\MonF.

\end{lemma}
\begin{proof}
We have to prove soundness only for the new family of equations $\mathcal{O}$ as the other equations are sound by Theorem \ref{thm:GroundSound}. 

First of all, note that $\sigma(x  + s.x + \overline{s}^{(k)}(\yes + \no))$ accepts every trace accepted by $\sigma(x + \overline{s}^{(k)}(\yes + \no))$, and rejects every trace rejected by $\sigma(x + \overline{s}^{(k)}(\yes + \no))$. We are therefore left to show that
\begin{itemize}
\item if $ \sigma(s.x)$ accepts some trace then so does $\sigma(x + \overline{s}^{(k)}(\yes + \no))$, and
\item if $ \sigma(s.x)$  rejects some trace then so does $\sigma(x + \overline{s}^{(k)}(\yes + \no))$.
\end{itemize}
We only detail the proof for the latter claim, as that of the former one is similar. To this end, assume that $ \sigma(s.x)$  rejects some trace $s'$. Then $s'= ss'' $ for some $s''$ that is rejected by $ \sigma(s.x)$. If $s''$ is a prefix of $s^k$, then it is not hard to see that $\sigma(x)$ rejects $s'$ too, and thus so does $\sigma(x + \overline{s}^{(k)}(\yes + \no))$. On the other hand, if $s''$ is not a prefix of $s^k$, then  $s'=ss''$ is not a prefix of $s^k$ either. Therefore, $\sigma(\overline{s}^{(k)}(\yes + \no))$ rejects $s'$. It follows that $\sigma(x + \overline{s}^{(k)}(\yes + \no))$ rejects $s'$, and we are done. \end{proof}

%First of all we observe that for any substitution $\sigma$, we have that $\sigma(x + \overline{s}^{(k)}(\yes + \no))$ will be a summand of both sides of the equation and therefore all traces accepted and rejected by it under any $\sigma$ will also be common for both sides. The only part of the equation that is not trivially a summand of both sides is $\sigma(s.x)$. 

% 
%Fix a substitution $\sigma_0$ and a trace$s_0$ that the summand $\sigma_0(s.x)$ rejects. We prove that $s.s_0$ is rejected by the common summands of the equation. The case of a trace $s_0$ that is accepted $\sigma_0(s_0)$ is similar. Therefore, both sides of the equation accept and reject the same trace which suffices the complete the proof. Also if $\sigma_0(x)$ contains a $no$ summand we would have the trivial rejection from both sides of the equation so, we can assume that is not the case. We have now that if $\sigma_0$ maps $x$ to a term $m_x$ such that $ s_0 \in L_r(m_x)$ with $s_0$ a prefix of $s$, then both sides reject $s_0$ and consequently $s.s_0$. On the other hand if $s_0$ is not a prefix of $s$ then the trace $s.s_0$ is also not a prefix of $s.s$ and therefore the summand $\overline{s}^{(k)}.(\yes +no)$ for any $k$ accepts and rejects the trace $s.s_0$ and therefore the summand $\sigma_0(s.m_x)$. 

We provide here some examples of how to use the above to derive some simpler and more intuitive sound equations. 
\begin{lemma}
The following equations are derivable from $\mathcal{O}$ for each $s,s_1 \in \Act^*$: 
\begin{enumerate}
	\item $x + s.x + s.\overline{s_0}(\yes +\no) = x +s.\overline{s_0}(\yes +\no)$, with $s_0$ a prefix of $s$,

	\item $\yes + x + s_1.\overline{s_2}(\no)  = \yes + x + s_1.\overline{s_2}(\no) + s_1.x$, where $s_2$ is any prefix of $s_1$, 
    \item $\no + x + s_1.\overline{s_2}(\yes)  = \no + x + s_1.\overline{s_2}(\yes) + s_1.x$, where $s_2$ is any prefix of $s_1$,   
	\item $x + s.\displaystyle\sum_{a \in \Act} a.(\no + \yes ) =x + s.(x +         \sum_{a \in \Act} a.(\no + \yes ))$.
\end{enumerate}

\end{lemma}

\begin{proof}
We first show how to derive the first equation and then we derive the rest from it. 
We start by picking the equation $O2_{s,1}$ i.e. 
$$x + s.x + s.\overline{s}^{\leq}(\yes +\no) + s.s.\displaystyle\sum_{a \in \Act} a.(\yes +\no) = $$
$$ x + s.\overline{s}^{\leq}(\yes +\no) + s.s.\displaystyle\sum_{a \in \Act} a.(\yes +\no)~.
$$ 

In addition we have the tautology $$s.s_0.\displaystyle\sum_{a \in \Act} a.(\yes+\no) = s.s_0.\displaystyle\sum_{a \in \Act} a.(\yes+\no)~ ,$$ for the specific prefix $s_0$ of $s$.  On the two valid above equations we apply the congruence rule for $+$ and have:

$$x + s.x + s.\overline{s}^{\leq}(\yes +\no) + s.s.\displaystyle\sum_{a \in \Act} a.(\yes +\no) +s.s_0.\displaystyle\sum_{a \in \Act} a.(\yes +\no)$$
$$ = x + s.\overline{s}^{\leq}(\yes +\no) + s.s.\displaystyle\sum_{a \in \Act} a.(\yes +\no) +s.s_0.\displaystyle\sum_{a \in \Act} a.(\yes +\no)~.
$$ 

The first simplification that we perform now is by observing that the summand $s.s_0.\displaystyle\sum_{a \in \Act} a.(\yes +\no)$ accepts and rejects a prefix of the whole summand $s.s.\displaystyle\sum_{a \in \Act} a.(\yes +\no)$ and therefore we can eliminate the latter from the summation: $$x + s.x + s.\overline{s}^{\leq}(\yes +\no)+s.s_0.\displaystyle\sum_{a \in \Act} a.(\yes +\no)$$
$$ = x + s.\overline{s}^{\leq}(\yes +\no) + s.s_0.\displaystyle\sum_{a \in \Act} a.(\yes +\no)~.$$
In addition the term $s.\overline{s}^{\leq}$ can be rewritten as $s.\overline{s_0}^{\leq}(\yes +\no)+ s.s_0.\overline{s_1}(\yes +\no)$ with $s = s_0.s_1$. To see this, consider that the traces up to length $|s|$ that do not cause a rejection of the trace $s$ are the ones  that do not cause a rejection of its prefix $s_0$ and the ones that start with $s_0$ but do not cause the rejection of its continuation $s_1$. Thus we have: 

$$x + s.x + s.\overline{s_0}^{\leq}(\yes +\no)+s.s_0.\displaystyle\sum_{a \in \Act} a.(\yes +\no) + s.s_0.\overline{s_1}(\yes +\no)$$
$$ = x + s.\overline{s_0}^{\leq}(\yes +\no)+s.s_0.\displaystyle\sum_{a \in \Act} a.(\yes +\no) + s.s_0.\overline{s_1}(\yes +\no)~.$$ 

Now we have again that the summand $s.s_0.\displaystyle\sum_{a \in \Act} a.(\yes +\no)$ accepts and rejects a prefix of the whole summand $s.s_0.\overline{s_1}(\yes +\no)$ and therefore we can omit the latter. This gives us the equation: 

$$x + s.x + s.\overline{s_0}^{\leq}(\yes +\no)+s.s_0.\displaystyle\sum_{a \in \Act} a.(\yes +\no)= $$ 
$$x + s.\overline{s_0}^{\leq}(\yes +\no)+s.s_0.\displaystyle\sum_{a \in \Act} a.(\yes +\no) ~,$$ which can be rewritten using our notation as $$x + s.x + s.\overline{s_0}(\yes +\no)= x + s.\overline{s_0}(\yes +\no)~,$$ giving us the target equation. 

Having presented the proof for the first family of equations in detail we give a short description for the rest. For the equations $(2)$ and $(3)$
it suffices to use the congruence rule for $+$ with the equations $\yes  = \yes$ and $\no = \no$ respectively and then simplify the equations by using the distribution axiom for $+$. For the latter equation $(4)$ it is enough to instantiate the prefix $s_0$ in the the family of equations $(1)$ as the empty string $\varepsilon$. This is, of course, allowed since the empty string is a prefix of any string. \end{proof}
%%%%%%%%%%

Now that we have discussed the family of axioms $\mathcal{O}$, we proceed to use them in defining a notion of reduced normal form that is suitable for monitors over a finite action set.
\begin{definition}\label{def:RNFopenfin}
A \textbf{\textit{finite-action-set reduced normal form}} is a term $$m = \displaystyle\sum_{a \in A} a. m_a + \sum_{i \in I} x_i~ [+\yes]~ [+\no]$$
where  each $m_a$ is different from $\vend$ and if $v \in \{ \yes,\no\}$ is a summand of m then each $m_a$ is $v$-free, and in reduced normal form. If both $\yes$ and $\no$ are summands of $m$ then $m$ is equal to $\yes + \no$. In addition for every trace $s$, if there exists a $k$ such that for all the traces $s_0$ $$\overline{s}^{(k)}(\yes+\no) \xrightarrow[\text{}]{s_0} \yes +\no,~\text{implies:}~ m \xRightarrow[\text{}]{s_0} \yes ~\text{and}~ m \xRightarrow[\text{}]{s_0} \no$$ then $m \not \xrightarrow[\text{}]{s} x_i +m'$ for all $i \in I $ and $m'$.

\end{definition}

In order to use the above form of the monitors in \MonF~ we need to prove that any term  can be rewritten in a reduced normal form using the axioms in $\mathcal{E}_{\textit{v,f}}'$. Before doing so we will prove the following useful lemma, which only uses axioms form $\mathcal{E}_v$. 
\begin{lemma}\label{lem:seperationOfTraces}
For a monitor $m \in$ \MonF:
\begin{itemize}
\item if $m \xRightarrow[\text{}]{s} \yes $  then $\mathcal{E}_v \vdash m = m + s.\yes$ and
\item if $m \xRightarrow[\text{}]{s} \no$ then $\mathcal{E}_v \vdash m = m + s.\no$.
\end{itemize} 
\end{lemma}
\begin{proof}
We prove both statements by induction on the length of the trace $s$ and limit ourselves to presenting the proof for the first one. 
\begin{itemize}
\item If $s$ is the empty trace, then $m$ accepts the empty trace. Therefore it must contain a $\yes$ syntactic summand and we are done. 
% by Antonis: And the lemma follows form reperitive applications of the axio $Y_a$. 
\item Assume now that $s = a.s'$. Then $m \xrightarrow[]{a} m_a \xrightarrow[]{s'} \yes$ for some $m_a$. By induction $ \mathcal{E}_{\textit{v,f}}' \vdash m_a = m_a + s'.\yes$.

Now,  
%$$
\begin{align*}
\mathcal{E}_{\textit{v,f}}' \vdash m &= m + a.m_a= m + a.(m_a + s'.\yes)
%$$ $$ 
= m + a.m_a + a.s'.\yes\\ &= m +a.s'.\yes
\tag*{\qedhere}
\end{align*}
%$$

% we have proved that we can create a syntactic summand in $m$ for all of the traces up to length $\ell$. 
%\item Take now a trace $s_0$ of length $\ell + 1$.
%Since $m$ must accept and reject the trace $s_0$ it is necessary that $m \xrightarrow[\text{}]{s_0} \yes$. If $m \xrightarrow[\text{}]{s_{0[-1]}} \yes$ with $s_{0[-1]}$ being the trace of all but the last action of $s_0$ then we can use the induction hypothesis and say that $\mathcal{E}_{fin} \vdash  m= s_{0[-1]}.\yes + m $. Then we use axiom $Y_a$ for $a$ being the last action of $s_0$ to transform said equation to $\mathcal{E}_{fin} \vdash  m= s_{0[-1]}.(\yes + a.\yes) + m$ and then by the distribution axiom $D_a$ we have the conclusion. If $m \not \xrightarrow[\text{}]{s_{0[-1]}} \yes$ then it must be the case where $m$ can perform all actions of $s$ and arrive at a $\yes$. By applying the distribution axiom $D_a$ for each action of $s_0$ we have the necessary syntactic summand. 
\end{itemize} \end{proof}

We will also need a similar result, this time involving syntactic summands that contain occurrences of variables. 
\begin{lemma}\label{lem:seperationOfVariables}
For a monitor $m \in$ \MonF, where $m$ is in normal form for open terms, if $m \xrightarrow[\text{}]{s} x + m_s$ then $\mathcal{E}_{v} \vdash m = m' + s.x$ where $m' \not \xrightarrow[\text{}]{s} x + m''$ for every $m''$. 

\end{lemma}

\begin{proof}
We prove the claim by induction on the length of the trace $s$. 
\begin{itemize}
\item If $s$ is the empty trace then $m \xrightarrow[\text{}]{\varepsilon} x + m_s = m$. This means that $x$ is a summand of $m$. Since $m$ is in normal form, $m_s$ does not have $x$ as a summand and we are done.

\item Assume now that $s = a.s'$. Since m is in normal form and $m \xrightarrow[]{a.s'} s + m'$, we have that $m = m' + a.m_a$ for some $m' \not\xrightarrow[]{a}$ and $m_a$ in formal form such that $m_a  \xrightarrow[]{s'} x+m'$. By the induction hypothesis, $ \mathcal{E}_{v} \vdash m_a = m_a' + s'.x$ where $m_a' \not\xrightarrow[]{s'}x + m_{s'}'$ for every $m_{s'}'$.

Therefore we have:

$$m = m' + a.m_a = m + a.(m_a' + s'.x) = m' + a.m_a' +a.s'.x~,~~$$
and since $m_a' \not\xrightarrow[]{s'} x+ m_{s'}'$ for every $m_{s'}'$ and $m' \not\xrightarrow[]{a}$ we have that $ m = s.x + m_{rest}$ with $m_{rest} \not\xrightarrow[]{s} x + m''$ for every $m''$ and we are done. 
\qedhere
\end{itemize} \end{proof}

\begin{lemma}\label{lem:RNFOpenFin}
Each open monitor $m \in$ \MonF, is provably equal to a reduced normal form using $\mathcal{E}_{\textit{v,f}}'$.
\end{lemma}

\begin{proof}
From Lemma \ref{lem:NormalFormOpen} we can start from a monitor $m$ already in open normal form, as given in Definition \ref{def:NormalFormOpen}. Therefore we have the following cases: 

\begin{itemize}
\item $m = \yes + \no$.
\item $m = \yes + \displaystyle\sum_{a \in A} a.m_a + \sum_{i \in I} x_i$, where each $m_a$ is $\yes$-free.
\item $m = \no + \displaystyle\sum_{a \in A} a.m_a + \sum_{i \in I} x_i$, where each $m_a$ is $\no$-free.
\item $m = \displaystyle\sum_{a \in A} a.m_a + \sum_{i \in I} x_i$.
\end{itemize}
We begin our analysis from the second case. A similar analysis can be applied to the third one and the fourth one follows by a simpler version of the same inductive argument.
We have therefore a monitor $m  = \yes +  \displaystyle\sum_{a \in A} a.m_a + \sum_{i \in I} x_i$. The extra claim for these reduced normal forms is that if for some trace $s$, and a $k_0$, for all traces $s_0$, $$\overline{s}^{(k_0)}(\yes+\no) \xrightarrow[\text{}]{s_0} \yes +\no,~\text{implies:}~ m \xRightarrow[\text{}]{s_0} \yes ~\text{and} ~ m \xRightarrow[\text{}]{s_0} \no$$
then $m \not \xrightarrow[\text{}]{s} x_i +m_x$ for all $i \in I $ and $m_x$. In order to prove this extra constraint we assume the premise is true. We will show that we can reduce $m$ to $m_{red}$ with $\mathcal{E}_{fin } \vdash m = m_{red}$ and $m_{red} \not\xrightarrow[\text{}]{s} x_i + m_x $ for every $i \in I$ and every $m_x$.

Since $m$ accepts and rejects all the traces that $\overline{s}^{(k_0)}(\yes+\no)$ accepts and rejects, we have that $m \simeq m + m'$ and that $m' \xrightarrow[\text{}]{s_0} \yes + \no$ for all of the traces $s_0$ that $\overline{s}^{(k_0)}(\yes+\no) \xrightarrow[\text{}]{s_0} \yes +\no$. We call this set of traces $\mathcal{S}$ which is finite since $k_0$ is fixed. Therefore by Lemma \ref{lem:seperationOfTraces} we have that $\mathcal{E}_{\textit{v,f}}' \vdash m = m + \displaystyle\sum_{s_0 \in \mathcal{S}} s_0.(\yes +\no)$. Since the term $ \displaystyle\sum_{s_0 \in \mathcal{S}} s_0.(\yes +\no)$ is verdict equivalent to $\overline{s}^{(k_0)}(\yes+\no)$ and both terms are closed, we have that by Theorem \ref{thm:GrCompFin}, $\mathcal{E}_v \vdash \displaystyle\sum_{s_0 \in \mathcal{S}} s_0.(\yes +\no) = \overline{s}^{(k_0)}(\yes+\no)$. Therefore $\mathcal{E}_{\textit{v,f}}' \vdash m = m + \overline{s}^{(k_0)}(\yes+\no)$, which means
 $$\mathcal{E}_{\textit{v,f}}' \vdash m = \yes +  \displaystyle\sum_{a \in A} a.m_a + \sum_{i\in I} x_i + \overline{s}^{(k_0)}(\yes+\no).$$ 

For the same monitor $m$ we now want to argue that if $m\xrightarrow[\text{}]{s} x_i$ for one of the variables in $\{ x_i \mid i \in I\}$ then we can eliminate this occurrence.

Since $m$ is in reduced normal form we have by Lemma \ref{lem:seperationOfVariables} that $m = m' + s.x_i$ where $m' \not\xrightarrow[]{s} x_i$. Additionally we have shown that $m = m + \overline{s}^{(k_0)}(\yes+\no)$ which implies $m = m' + s.x_i + \overline{s}^{(k_0)}(\yes+\no)$ with $m' \not\xrightarrow[\text{}]{s} x_i $. Since $x_i$ is one of the variables that appear as summands of $m$ we can successfully apply the axiom $O2_{s,k_0}$ for each variable and we have the that indeed $m$ reduces to a monitor $m_{red}$ such that $m_{red} \not\xrightarrow[\text{}]{s} x_i + m_{x_i}$ for every $i \in I$ and every $m_{x_i}$. \end{proof}

\begin{lemma}\label{lem:ExistsTrace}
If monitor $m \in$ \MonF, with $|\Act| \geq 2$ is in reduced normal form and contains an $x$ summand and $m \xrightarrow[\text{}]{s} x + m'$ for some $m'$ then there is at least one trace $s_{bad}$ such that  for every $k$, $$ \overline{s}^{(k)}(\yes + \no)\xRightarrow[\text{}]{s_{bad}}\yes~\text{and}~\overline{s}^{(k)}(\yes + \no)\xRightarrow[\text{}]{s_{bad}} \no$$ but $$m \not \xRightarrow[\text{}]{s_{bad}} \yes~\text{or}~m \not \xRightarrow[\text{}]{s_{bad}} \no.$$
\end{lemma}

\begin{proof}

We can easily show that for each $k$ there exists an $s_k$ such that  $\overline{s}^{(k)}(\yes + \no)\xrightarrow[\text{}]{s_k} \yes +\no $ but $m \not \xRightarrow[\text{}]{s_k} \yes$ or $m \not \xRightarrow[\text{}]{s_k} \no$. This follows since if this were not the case then for some $k_0$, no such trace $s_{k_0}$ exists. Thus the monitor would contain a summand $m' \simeq \overline{s}^{(k_0)}(\yes + \no)$ for this $k_0$ and still it would be able to perform the transition $m \xrightarrow[\text{}]{s} x +m'$ which contradicts the assumption that $m$ is in reduced normal form.

We will now show that one trace $s_{bad}$ suffices for all $k$. To that end, consider the term $\overline{s}^{(1)}(\yes + \no)$. If there is an $s_1$, which is not a prefix of $ss$ and $m \not \xRightarrow[\text{}]{s_{1}} \yes$ or $m \not \xRightarrow[\text{}]{s_{1}} \no$ then for $s_{bad} = s_1$ we have that for all $k$, $\overline{s}^{(k)}(\yes + \no)\xRightarrow[\text{}]{s_{bad}}\yes~\text{and}~\overline{s}^{(k)}(\yes + \no)\xRightarrow[\text{}]{s_{bad}} \no $ and we are done. If this is not the case and since the trace $s_1$ is guaranteed to exist (by the previous paragraph) then it must be an extension of $ss$. Again if $s_1$ is not prefix of $sss$ then again for $s_{bad} =s_1 $ we have the necessary conclusion.

 Otherwise $m \not \xRightarrow[\text{}]{s_{1}} \yes$ or $m \not \xRightarrow[\text{}]{s_{1}} \no$ for the trace $s_1 = ssa$, where $a$ is the first action of $s$. Therefore by the definition of $\overline{s}^{(k)}(\yes + \no)$ we have that for all $s_b$ such that $\overline{s}^{(k)}(\yes + \no)\xrightarrow[\text{}]{s_b}\yes +\no$ and $k > 1$ we have that $m \not \xRightarrow[\text{}]{s_b} \yes$ or $m \not \xRightarrow[\text{}]{s_b} \no$. This allows us to look for an $s_{bad}$ which will also cover the case $k=1$ in larger terms. 

We then apply the same reasoning for $k =2 , \ldots $ up to a certain $k_b$. If at any point in the process we encounter a trace $s_i$ which fulfills our premise then we can stop. We are just left to show that this process will eventually terminate.

 This can be shown as follows. Recall that every monitor $m$ has a finite depth $\depth(m)$ (see Def.~\ref{def:depth}). We now take a $k_{b}$ large enough so that $s^{k_{b}} > \depth (m)$. If the iterative procedure described above reaches this $k_b$ we have that $m \not \xRightarrow[\text{}]{s_{bad}} \yes$ or $m \not \xRightarrow[\text{}]{s_{bad}} \no$ for the trace $s^{k_b +1}a$ where $a$ is the first action of $s$. However since the depth of the monitor $m$ is smaller that the length of this trace we also have that the monitor cannot accept or reject any of its extensions.

Therefore for the extension $s_{bad} = s^{k_b +1}ac$ where $c$ is not the second action of $s$ we have that for all $k > k_b$, $\overline{s}^{(k)}(\yes + \no)\xRightarrow[\text{}]{s_{bad}}\yes~\text{and}~\overline{s}^{(k)}(\yes + \no)\xRightarrow[\text{}]{s_{bad}} \no$ while $m \\not \xRightarrow[\text{}]{s_{bad}} \yes$ or $m \not \xRightarrow[\text{}]{s_{bad}} \no$. Additionally since the iterative procedure we described above reached this $k_b$ we have that for all $ i \leq j \leq k_b$, it is true that $\overline{s}^{(k)}(\yes + \no)\xRightarrow[\text{}]{s_{j}}\yes~\text{and}~\overline{s}^{(k)}(\yes + \no)\xRightarrow[\text{}]{s_{j}} \no$, which concludes the proof. \end{proof}

%%%%%%%%%%%%%%%%%%%%%%%%

The two lemmata above play a key role in the completeness proof we
will present now.

 We distinguish two cases separately, namely when $|\Act| \geq 2$ and when $\Act$ is a singleton. This is necessary because equations such as $x = x + a.x$ are only sound when $\Act = \{ a\} $. For the proof when $|\Act| \geq 2$ it is necessary to utilize at least two actions $a,b \in \Act$, which is the reason why when only one action is available new cases arise.
\paragraph{\textbf{Action set with at least two actions}}

We have already shown the soundness of the axiom system $\mathcal{E}_{\textit{v,f}}'$. We now proceed to show completeness. 

For each such completeness theorem we follow a similar general strategy in order to prove that two arbitrary verdict equivalent monitors have identical reduced normal forms.  To that end, we prove that they have identical variables as summands, that the sets of initial actions that each one can perform are equal and that after a common action they reach monitors that are also verdict equivalent. Unfortunately, for a finite set of actions, we were not able to define a substitution that would cover all the three above-mentioned steps like we did when the set of actions was infinite. We therefore adopted a proof strategy that focuses on each part of the proof separately. 
 
\begin{theorem}\label{thm:CompOpenFin}
$\mathcal{E}_{\textit{v,f}}'$ is complete for open terms for finite Act with $|\Act| \geq 2$. That is, if $m \simeq n$ then $\mathcal{E}_{\textit{v,f}}' \vdash m = n$.
\end{theorem}

\begin{proof}

By Lemma \ref{lem:RNFOpenFin} we may assume that $m$ and $m$ are in reduced normal form. We prove the claim by induction on the sum of the sizes of $m$ and $n$, and proceed with a case analysis on the form $m$ may have.

In the case where $m$ contains both a $\yes$ and a $\no$ summand then both $m$ and $n$ must be equal to $\yes + \no$ as they are in reduced normal form. 

Assume now that $$m = \yes+ \displaystyle\sum_{a \in A} a.m_a + \sum_{i \in I} x_i~, $$ where $\{x_i \mid i \in I \}$ is the set of variables occurring as summands of $m$ and each $m_a$ is $\yes$-free and different from $\vend$ (as a reduced normal form). Since $\sigma(m)$ accepts $\varepsilon$ for each $\sigma$ and $m \simeq n$, monitor $n$ is bound to have a similar form since it must contain the verdict $\yes$ as a summand (but not a $\no$ one). Therefore: $$n = \yes+ \displaystyle\sum_{b \in B} b.n_b + \sum_{j \in J} y_j$$ and we need to show that there is a way to apply our axioms to show that monitor $n$ is provably equal to $m$.

We start by proving that $\{x_i \mid i \in I \} = \{y_j \mid j \in J \}$. By symmetry, it suffices to show that $\{x_i \mid i \in I \} \subseteq \{y_j \mid j \in J \}$. To this end, assume $x \in \{x_i \mid i \in I \}$. Consider the substitution $\sigma$ mapping $x$ to $\no$ and every other variable to $\vend$, i.e:
\[
\sigma(y) = 
\begin{cases}
\no, & \textit{if}~ y = x \\
  \vend,& \textit{otherwise}.
\end{cases}
\] 

Then, $\sigma(m)$ rejects the empty trace $\varepsilon$. Since $\sigma(m) \simeq \sigma(n)$, we have that $\sigma(n)$ must also reject $\varepsilon$. By the form of $n$ and the definition of $\sigma$, this is only possible if $n$ has $x$ as a summand, and we are done. Therefore the set of variables of $m$ is a subset of the variables of $n$. 

%By constructing the symmetric substitution, the set of variables of $n$ is proved a subset of the variables of $m$ which makes them identical. 

Next, we  prove that the action sets $A,B$ are identical. 
%If both $A$ and $B$ are empty then they are equal and the monitors at also trivially proven equal. Assume now that $A$ is non empty. 
%At least one of the sets $A,B$ must be non-empty otherwise they are trivially equal.
%Since $m$ is not equal to $\yes +no$ there is at least one $a \in A$. 
Assume that $ a \in A$. Since $\Act$ contains at least two actions, there is some action $b \neq a$. Consider the substitution $\sigma_1$ defined by $\sigma_1(x) = b.\no$ for each $x \in \Var$. Since $a \in A$ and $m_a$ is $\yes$-free and different from $\vend$, it is easy to see that there exists an $s \in \Act^*$ such that $as\in L_r(\sigma_1(m))$. Since $m \simeq n$ we have that $\sigma_1(m) \simeq \sigma_1(n)$ and therefore $\sigma_1(n)$ must also reject $as$.  By the form of $n$ and the definition of $\sigma$, this is only possible if $n \xrightarrow[]{a} n_a$ for some $n_a$ and therefore $a \in B$. Hence, $A \subseteq B$ and the claim follows by symmetry.

For the final part of the proof we must show that $m_a \simeq n_a $ for each $a \in A$, which is enough to complete the proof, by the induction hypothesis.
Towards a contradiction we will assume that the two monitors $m_a,n_a$ are not verdict equivalent. Therefore there exists a substitution $\sigma_0$ that separates them, that is without loss of generality, there is a trace $s_0$ such that $s_0\in L_r(\sigma_0(m_a)), s_0 \not\in L_r(\sigma_0(n_a))$ or there is some $s_0 \in L_a(\sigma_0(m_a)), s_0 \not\in L_a(\sigma_0(n_a))$ .

%We are aware that if $m_a$ is not verdict equivalent to $n_a$ then various substitutions would cause different sorts of disagreements. For our convenience we will later on pick a specific one. As for now we continue without assuming something about the substitution. 

We will analyze first the case of rejection of the string $s_0$.  The substitution $\sigma_0$ must be a closed one for $m_a , n_a$ i.e. it must map to a closed monitor all variables in $(Var(m_a)\cup Var(n_a))$. We will use this substitution to create a new one $\sigma_{bad}$ that would also separate the original monitors $m,n$.

%The first step towards this is: $\sigma_{bad} (Var(m)\setminus (Var(m_a)\cup Var(n_a))) = b.\no$, where $b\in \Act$ and $b\not= a$. Also  $\sigma_{bad} ((Var(m_a)\cup Var(n_a))) = \sigma_0 $. It is important to see here that in order for the above argument to work there must be at least two discrete action in $|\Act|$. 
The first step towards this is:

$$
\sigma_{bad} (x) =
\begin{cases}
\vend ,~\text{if}~ x \in Var(m)\setminus (Var(m_a)\cup Var(n_a)),\\
  \sigma_0(x),~ otherwise. \\
\end{cases}$$ 

Now since $s_0 \in L_r(\sigma_0(m_a))$ and $\sigma_{bad}(m_a) = \sigma_{0}(m_a)$ we also know that $a.s_0 \in L_r(\sigma_{bad}(a.m_a))$. Our aim is to show that  $a.s_0 \not\in L_r(\sigma_{bad}(n))$. Following the definition $\sigma_{bad}(n_a) = \sigma_0 (n_a)$ and therefore $s_0 \not\in L_r(\sigma_{bad}(n_a))$.

Hence, the only way for $\sigma_{bad}(n)$ to reject $a.s_0$, like $\sigma_{bad}(m)$ does, is if it was rejected by the mapping of one the variables contained in the set $\{ x_i \mid i \in I  \}$.

It is useful to make here apparent that in order for $\sigma_{bad}(n)$ to reject $a.s_0$, it must do so completely independently of the summand $\sigma_{bad}(a.n_a)$, since the latter cannot reject any of the prefixes of $a.s_0$ as well. 
Even in the case where $s_0$ starts with $a$, and $\sigma_0(n_a)$ rejects some $a.s_1.s_2.\ldots s_{n-i}$ it would still be impossible for $\sigma_0(n_0)$ to reject $a.s_0$ since the assumption that $a.s_0 = a.a.s_1.s_2. \ldots s_{n-1}$ would automatically imply that $\sigma_0(n_a)$ rejects some prefix of $s_0$ which is a contradiction. 

%\begin{figure*}[h]
%\begin{tikzpicture}
%  [
%    grow                    = right,
%    sibling distance        = 6em,
%    level distance          = 10em,
%    edge from parent/.style = {draw, -latex},
%    every node/.style       = {font=\footnotesize},
%    sloped
%  ]
%  \node [root] {$m$}
%    child { node [env] {$m_a$}
%      edge from parent node [above] {a} }
%    child { node [dummy] {}
%      child { node [dummy] {}
%        child { node [env] {align\\flalign}
%          edge from parent node [below] {at relation sign?} }
%        child { node [env] {alignat}
%          edge from parent node [above] {at several}
%                           node [below] {places?} }
%        child { node [env] {gather}
%                edge from parent node [above] {centered?} }
%        edge from parent node [below] {aligned?} }
%      child { node [env] {multline}
%              edge from parent node [above, align=center]
%                {first left,\\centered,}
%              node [below] {last right}}
%              edge from parent node [above] {multi-line?} };
%\end{tikzpicture}
%\end{figure*}
\begin{figure}[h]
\begin{tikzpicture}
  [
    grow                    = right,
    sibling distance        = 6em,
    level distance          = 9em,
    edge from parent/.style = {draw, -latex},
    every node/.style       = {font=\footnotesize},
    sloped
  ]
  \node [env] {$\sigma_{bad}(m)$}
    child { node [dummy] {$\no$}
      edge from parent node [above] {$s''$} }
    child { node [env] {$\sigma_{bad}(m_a)$}
      child { node [env] {$\sigma_{bad}(m_s +x)$}
        child { node [dummy] {$\no$}
          edge from parent node [above] {$s''$} }
        edge from parent node [above] {$s'$} }
              edge from parent node [above] {$a$} };
\end{tikzpicture}
\caption{Transitions the monitor $\sigma_{bad}(m)$ can perform}
\end{figure}
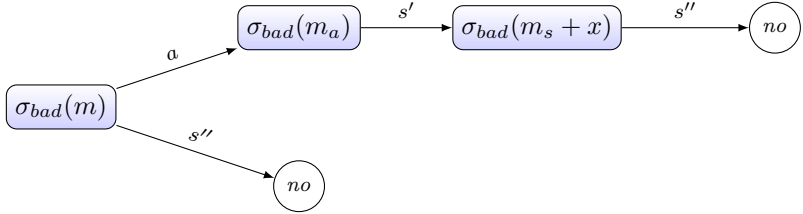

By the definition of $\sigma_{bad}$, the variables that did not appear at all in $n_a$ or $m_a$ were mapped to $\vend$ and therefore cannot reject any string. Therefore the only way for $n$ to reject $a.s_0$ is for one of the variables appearing in $Var(n_a) \cup Var(m_a)$ to have been mapped to a closed term that can reject $a.s_0$. (Note that this does not contradict the fact that $\sigma_{bad}(n_a)$ does not reject $s_0$). Therefore there is at least one $x_0 \in Var(m_a) \cup Var(n_a)$ and $x_0 \in \{x_i \mid i \in I \}$ such that $as_0 \in L_r(\sigma_{bad}(x_0) )$.

This leads to the case were $m,n$ reject a prefix of $as_0$ because of the mapping of $x_0$.
%($|\sigma_{bad}(x_0)| \geq 1$)
However this implies that we have the following situation: $$m = \yes + x_0 + a.m_a + \displaystyle\sum_{b \in A\setminus \{a\}} b.m_b + \sum_{i \in I\setminus\{0\}} x_i \simeq $$

$$\yes + x_0 + a.n_a +\displaystyle\sum_{b \in A\setminus \{a\}} b.n_b + \sum_{i \in I\setminus\{0\}} x_i = n$$  
 and that the monitor $m_a$ can perform the transitions: $ m_a \xrightarrow[\text{}]{s'} m_a' + x_0 $  and the monitor $\sigma_0(x_0) = \sigma_{bad}(x_0)$ respectively can perform the transitions: 
$ \sigma_{bad}(x_0) \xrightarrow[\text{}]{s''} \no ~,$ where $s'$ is a prefix of $s_0$ (i.e. $s_0 = s'.s''$) and in addition $n_a \not \xrightarrow[]{s'} x+n'$ for any $n'$. This means respectively that $ m \xrightarrow[]{as'}  m_a' + x_0$ and $ \sigma_{bad}(m_a' + x_0) \xrightarrow[\text{}]{s''} \no ~.$

By Lemma \ref{lem:ExistsTrace} we have that there exists at least one trace $s_b$ such that $m \not \xRightarrow[\text{}]{s_b} \yes $ or $m \not \xRightarrow[\text{}]{s_b} \no$ but $s_b \in L_r(\overline{as'}^{(k)}(\yes+\no))$ for all $k\geq 0$. Since $m$ contains a $\yes$ summand we have that it must be the case that  $m \not \xRightarrow[\text{}]{s_b} \no$. We now, further modify $\sigma_{bad}$ to map the variable $x_0$ to $s_b.\no$ and any other variable $y \neq x_0$ to $\vend$. We have then that $s_b$ and $as'.s_b \in L_r(\sigma_{bad}(m))$. In addition $s_b \in L_r(\sigma_{bad}(n))$. However the traces that are rejected by the term $\overline{as'}^{(k)}$, by definition, are exactly the traces such that their rejection does not cause a rejection of the $as'$ trace. This means that under the modified substitution $\sigma_{bad}$, monitor $n$ \emph{cannot} reject the trace $as'.s_b$. This deems the monitors $m,n$ not verdict equivalent, which contradicts our assumption. We conclude then that the rejection set of $m_a$ is equal to the rejection set of $n_a$ for each $a \in A$.

It remains to show that $m_a$ and $n_a$ also have identical acceptance sets. Towards a contradiction, assume they do not and take a trace $s$ that under some substitution $\sigma_0$ separates them, i.e. $s \in L_a(\sigma_0(m_a))$ and $s \not\in L_a(\sigma_0(n_a))$. In addition, assume that $s$ is of minimum length, meaning that no prefix of $s$ (under any substitution) has the property of separating the acceptance sets of $m_a$ and $n_a$. This fact in addition to $m_a$ and $n_a$ being $\yes$-free (as a result of $m$ and $n$ being in reduced normal form) means that the acceptance of $s$ by $m_a$ is the result of a variable $x$ occurring in $m_a$ as $m_a \xrightarrow[\text{}]{s} x+m'$ for some $m'$. Since however the assumption is that $s \not\in L_a(\sigma_0(n_a))$ we have that $n_a \not \xrightarrow[\text{}]{s} x+n'$ for any $n'$.
We know that this is exactly the case since if the variable $x$ occurred earlier in $m_a$ then by mapping it to $\yes$ we would have a shorter trace being accepted by $\sigma_0(m_a)$ but not $\sigma_0(n_a)$.

We are sure now that monitor $\sigma_0(n_a)$ cannot perform the transition $ \sigma_0(n_a )\xRightarrow[\text{}]{~s~} \yes$, which means that not only it does not arrive at the variable $x$ after reading the trace $s$, but also does not arrive to the $\yes$ verdict for any of its prefixes (say $s'$) as that would imply that it can reach the $\yes$ verdict for $s$ as well.

Finally, by $n$ being in reduced normal form, and by $m_a$ not arriving at a $\no$ verdict for any of the prefixes $s'$ of $s$ (as this would mean that if becomes a $\no$ and therefore cannot perform the transitions $m_a \xrightarrow[\text{}]{~s~} x$ ) we know that $n_a$ does not arrive to the $\no$ verdict after reading the trace $s$ or any of its prefixes either.  
%
% Note here that any variable mapped to $no$ immediately makes any sequence following it both accepted and rejected.
 
Given all of the above we can now construct the substitution $\sigma_{bad}$ that would separate the rejection sets of $n_a,m_a$ which is enough to prove the contradiction as the case where such a substitution exists and separates the rejection sets of the two sub-monitors has already been covered.
The situation we have at hand is as follows: 

Monitor $\sigma_0(m_a)$ can arrive to the verdict $\yes$ after reading the trace $s$ while $\sigma_0(n_a)$ cannot and also neither $n_a$ nor $m_a$ can produce a $\no$ verdict for the trace $s$. Therefore if we switch the mapping of $x$ to $\no$ in $\sigma'$ and the verdicts of all other variables that where mapped to a $\no$ verdict to $\vend$ we have produced a substitution that causes $s$ to be rejected by $\sigma'(m_a)$ but not from $\sigma'(n_a)$. By utilizing our previous construction there exists another one that separates the monitors $n,m$ as well which is a contradiction.

We have concluded then that the $L_a(m_a) = L_a(n_a)$ and $L_r(m_a) = L_r(n_a)$ which means that they are verdict equivalent. Therefore we can apply the inductive hypothesis and have that $\mathcal{E}_{\textit{v,f}}' \vdash m_a = n_a$. Using now congruence rules we have that $\mathcal{E}_{\textit{v,f}}' \vdash m = n$.  
All other possible forms of monitors $m,n$ are sub-cases that the relative analysis can be applied symmetrically and therefore they are omitted. \end{proof}   

\paragraph{\textbf{Singleton Action Set}}

We proceed now with the analysis of the completeness result when $\Act = \{ a \}$.

 As we mentioned earlier, when $a$ is the only action, the equation 
\[
(V_1) \quad x = x + a.x
\]
is sound, but cannot be proved from the equations in $\mathcal{E}_{\textit{v,f}}'$ over $\{ a \}$. Indeed, unlike $V_1$, all the equations in $E_v$ are sound regardless of the cardinality of the action set and those in the family $\mathcal{O}$ introduce subterms of the form $\yes + \no$, which can never be removed in equational derivations.

\begin{theorem}\label{thm:CompOpenUnary}
The finite axiom system $\mathcal{E}_{v,1}'=\mathcal{E}_v' \cup \{ V_1\} $ is
complete for verdict equivalence over open monitors when $\Act
= \{a \}$. That is, if $m \simeq n$ then $\mathcal{E}_{v,1}' \vdash m
= n$. Hence, verdict equivalence is finitely based when $\Act
= \{a \}$.
\end{theorem}

\begin{proof}

Before we start the main proof we note that the new axiom $ V_{1}$ can prove the equation $x =a^n.x +x$ for each $n \geq 0$. 
%%%%%%%%%%%%
This is done as follows: if $n = 0$ then this is the axiom $A3$. Assume we can prove that equation for $n$. Then we can show it for $n+1$ thus: $$x \overset{\mathrm{V_1}}{=} x +a.x \overset{\mathrm{I.H.}}{=} x + a. (a^{n}.x +x) \overset{\mathrm{D_a}}{=} x + a.x + a ^{n +1}.x \overset{\mathrm{V_1}}{=} x + a ^{n +1}.x~~.$$

Note here that this means that $\mathcal{E}_v' \cup \{ V_1\}$ proves all the equations in $\mathcal{O}$ over $\{a\}$, which means that even though $\mathcal{E}_v' \cup \{ V_1\}$ is finite, it can prove the infinite family $\mathcal{E}_{\textit{v,f}}'$ over $\{a\}$. 

Let $m \simeq n$. By Lemma \ref{lem:RNFOpenFin}, we can assume that $m$ and $n$ are in reduced normal form. 
We will present the argument only for the case where $m = \yes + a.m_a + \displaystyle\sum_{i \in I} x_i$, where each $m_a$ is $\yes$-free, as every other case is either trivial or a sub-case of this one. 

By following the reasoning of previous proofs, we have that $n = \yes ~[+ a.n_a ] + \displaystyle\sum_{i \in I} x_i$.

Let us first consider the case that $a.n_a$ is not a summand of
$n$. (Note that this is possible, as witnessed by axiom $V_1$.)  That
is $$m = \yes + a.m_a + \displaystyle\sum_{i \in I} x_i \simeq \yes
+\displaystyle\sum_{i \in I} x_i = n \enspace .$$

Observe that, for each $s \in \Act^*$, we have $ m_a \not\xRightarrow[]{s} \no$. Indeed, $ m_a\xRightarrow[]{s} \no$ would imply that $m$ and $n$ are not verdict equivalent under the substitution $\sigma_{\vend}(x) = \vend$ for all $x$. This means that $m_a$ is both $\yes$- and $\no$-free. Moreover, note that the set of variables occurring in $m_a$ is included in  $\{x_i \mid i \in I\}$. To see this, assume that $x$ occurs in $m_a$, but is not contained in $\{x_i \mid i \in I\}$. Consider the substitution that maps $x$ to $\no$ and all the other variables to $\vend$. Again, we have that $m$ rejects some trace starting with $a$ while $n$ cannot reject any trace, which contradicts our assumption that $m \simeq n$.

%We can therefore see that $m_a$ is only able to perform the transitions
%$ m_a \xrightarrow[]{s} x + m_{s,x}$ for some $m_{s,x}$.
For each monitor $m'$, we define $\mathcal{V}(m')$ as the set of pairs $(s,x)$ such that $ m' \xrightarrow[]{s} x + m''$ for some $m''$. By structural induction on $m'$ and Lemma \ref{lem:seperationOfVariables}, one can easily prove that, when 
$m'$ is $\yes$- and $\no$-free, $\mathcal{E}_v$ proves $m' = \displaystyle\sum_{(s,x) \in \mathcal{V}(m')} {s.x }$.

Therefore  $m = \yes + a.\displaystyle\sum_{(s,x) \in \mathcal{V}(m_a)} s.x + \displaystyle\sum_{i \in I} x_i$. Since the only available action in $\Act$ is $a$ and the variables occurring in $m_a$ also occur in $\{x_i \mid i \in I\}$, we have that by applying the equations we proved earlier by using axiom $V_1$ we can prove $m = \yes + \displaystyle\sum_{i \in I} x_i = n$, and we are done.

%Our claim is that $m_a$ is equal to $a^k.x$ for some $k$. Otherwise,\\ $m_a = a^k.no $ for some $k$. Then by the substitution $\sigma(X) = a^{2*k}.\no$ the two monitors are proven not verdict equivalent. Therefore $m$ must be of the form $\yes + a^*.X +X$ which by the use of our new axioms is proved equal to $n = \yes +X$.

Assume now that $a.n_a$ is a summand of
$n$. We proceed to prove that  that $m_a \simeq n_a$.
In this case we have 

$$m_a = \displaystyle\sum_{(s,x) \in \mathcal{V}(m_a)} {s.x}  ~[+ a^h.\no]~\text{and}~
n_a = \sum_{(s,x) \in \mathcal{V}(n_a)} {s.x} ~[+ a^k.\no]~,$$
for some $h,k$. 

By mapping all variables to $\vend$ we can see that $h = k$. Additionally, for each variable $s$ and by using the axiom $V_1$ we can reduce both of the above summations so that only the shortest $s$ leading to $x$ is kept. By Lemma \ref{lem:firstvar}, we have that,  for each variable, this $s$ is identical for both sides of the equality $m \simeq n$ and we are done. \end{proof}

\subsubsection{Completeness of $\omega$-verdict equivalence}
 This section presents a complete axiomatization for $\omega$-verdict equivalence over \MonF. We have already presented the necessary axioms that capture $\omega$-verdict equivalence over closed terms, as well as the necessary ones to capture equivalence of terms that include variables. We will show here that the combination of the two axiom systems is enough for completeness of $\omega$-verdict equivalence over open terms and there is no need for extra axioms to be added. First we look at the case for a singleton action set, i.e. $\Act=\{a\}$.
In this case, the equation
$$(\mathbf{V1_{\omega}}) ~ x = a.x $$
is sound and we therefore we can shrink the axiom system to: 

$$\mathcal{E}_{\omega,1}' =  \{A1-A4\} \cup \{V1_{\omega}\} \cup \{O1\},$$
for which we prove: 

\begin{theorem}\label{thm:compOpenOmegaUnary}
$\mathcal{E}_{\omega,1}'$ is complete for $\omega $-verdict equivalence for open terms for a finite $\Act$, with $|\Act| =1$. That is, if $m \simeq_{\omega} n$ then $\mathcal{E}_{\omega,1}' \vdash m = n$. 
\end{theorem}
\begin{proof}
The proof of the above follows easily since, by using those equations, every term can be proved equal to one of the form $\displaystyle\sum_{i \in I} x_i~[+\yes]~[+\no]$, where $I$ is empty if both $\yes$ and $\no$ are summands, and two terms of that form are $\omega$-verdict equivalent iff they are equal modulo $A1-A4$. Note that, in this case, there are only four congruence classes of terms, namely the ones asymptotically  equivalent to subsets of $\yes$ and $\no$, so the quotient algebra is very small and equationally well behaved. \end{proof}

In the case where there the action set contains more than one action but is still finite we have a more interesting situation. We therefore define: $$\mathcal{E}_{\omega,f}' = \mathcal{E}_{\omega} \cup \mathcal{E}_{\textit{v,f}}'~,$$
for which we prove: 

\begin{theorem}\label{thm:CompOpenOmega}
$\mathcal{E}_{\omega,f}'$ is complete for $\omega $-verdict equivalence over open terms when $\Act$ is finite and $|\Act| \geq 2$. That is, if $m \simeq_{\omega} n$ then $\mathcal{E}_{\omega,f}' \vdash m = n$. 
\end{theorem}

The rest of this section is devoted to the proof of the above theorem. We start by showing a lemma that tells us that if two
monitors are $\omega$-verdict equivalent then they can only disagree on finitely many finite traces.

\begin{lemma}\label{lem:verdictOmegaVerdTracesFin}
For two monitors in \MonF, we have that $m \simeq_{\omega} n$ if and only if, for any substitution $\sigma$, the set 

\begin{gather*}
   \mathcal{S}_{m,n,\sigma} = \left(L_a(\sigma(m))\setminus  L_a(\sigma(n))\right) \cup ( L_r(\sigma(m))\setminus L_r(\sigma(n))) \\ 
    \cup \left(L_a(\sigma(n))\setminus  L_a(\sigma(m))\right) \cup ( L_r(\sigma(n))\setminus L_r(\sigma(m)))
\end{gather*}
is finite.
\end{lemma}

\begin{proof}

We prove both implications separately by establishing their contrapositive statements. For the implication from left to right, assume that  $\mathcal{S}_{m,n,\sigma}$ is infinite. It follows that there are some $\sigma$ and trace $s$ such that $s \in \mathcal{S}_{m,n,\sigma}$  with $$|s| > max\{\depth(\sigma(m)), \depth(\sigma(n))\}.$$

Assume, without loss of generality, that $\sigma(m)$ accepts $s$, but $\sigma(n)$ does not. Let $a \in \Act$. Then $sa^{\omega}$ is in $L_a(\sigma(m))\cdot \Act^{\omega}$. We claim that  $sa^{\omega}$ is not in $L_a(\sigma(n))\cdot \Act^{\omega}$. Indeed, $\sigma(n)$ does not accept any prefix of $s$, since it does not accept $s$ itself, and it does not accept $sa^i$ for any $i\geq 0$ because $|s| > \depth(\sigma(n))$. 
For the implication from right to left, assume, without loss of generality, that there are some substitution sigma and some $t\in \Act^\{omega$ such that $t$ is in $L_a(\sigma(m))\cdot \Act^{\omega}$, but not in $L_a(\sigma(n))\cdot \Act^{\omega}$. Since  $t$ is in $L_a(\sigma(m))\cdot \Act^{\omega}$, we have that there are some $s\in L_a(\sigma(m))$ and $u$ in $\Act^{\omega}$ such that $t = su$. It follows that $ss' \in L_a(\sigma(m))$  for each finite prefix $s'$ of $u$, but none of the $ss'$ is contained in $L_a(\sigma(n))$. Therefore,  $\mathcal{S}_{m,n,\sigma}$ is infinite, and we are done. \end{proof}

We are now ready to present the proof of the main theorem of this section (Theorem \ref{thm:CompOpenOmega}).
\begin{proof}

By Lemma \ref{lem:RNFOpenFin} we may assume without loss of generality that the monitors $m$ and $n$ are in finite-action-set reduced normal form (Definition \ref{def:RNFopenfin}).
% I.e. $$m = \displaystyle\sum_{a \in A} a.m_a + \sum_{i \in I} x_i ~[+\yes]~[+\no] ~\text{and}~n = \displaystyle\sum_{b \in B} b.n_b + \sum_{j \in J} y_i ~[+\yes]~[+\no]$$ .

 We proceed by a case analysis on the form $m$ and $n$ might have and by induction on the sum of the sizes of $m$ and $n$. 
\begin{itemize}
    \item Assume that $m = \yes +\no \simeq_{\omega} \displaystyle\sum_{a \in A}a.n_a  + \sum_{j \in J} y_j = n$. First of all, note that $A= \Act$. Indeed, assume $a \not\in A$. Then, under a substitution that maps every variable to $\vend$, all infinite traces starting with $a$ are neither accepted nor rejected by $n$ since $n$ cannot take an $a$ transition and it also does not accept and reject $\varepsilon$. However all infinite traces (including those starting from $a$) are both accepted and rejected by $m$, which is a contradiction as we have assumed that the two monitors are $\omega$-verdict equivalent.

Moreover, it is not hard to see that $n_a \simeq_{\omega} \yes +\no$ holds for each $a \in \Act$.
By the induction hypothesis,  $\mathcal{E}_{\omega,f}'$ proves $n_a = \yes + \no$,  for each $a \in \Act$.  Therefore,

$$\mathcal{E}_{\omega,f}' \vdash n = \displaystyle\sum_{a \in \Act} a.(\yes+\no) +  \sum_{j \in J} y_j \overset{\mathrm{D_a}}{=} \displaystyle\sum_{a \in \Act} a.\yes + \sum_{a\in \Act} a.\no + \sum_{j \in J} y_j$$

 $$ \overset{\mathrm{Y_{\omega},N_{\omega}}}{=} \yes +\no + \displaystyle\sum_{j \in J}  y_j\overset{\mathrm{O1}}{=} \yes +\no~~,$$ and we are done.

    \item Now, we assume that $m = \yes +\no \simeq_{\omega} \displaystyle\sum_{a \in A}a.n_a + \sum_{j \in J} y_j +\yes= n$, with each $n_a$ being $\yes$- and $\vend$-free. As above $A= \Act$. Moreover, for each $a \in \Act$, $L_r(n_a)\cdot \Act^{\omega}= \Act^{\omega}$. Following the same argument as above only for the $\no$ verdict we conclude that $$\mathcal{E}_{\omega,f}' \vdash n = \yes +\displaystyle\sum_{a \in \Act} a.\no + \sum_{j \in J}y_j =\yes +\no = m.$$

    \item The case  $m = \yes +\no \simeq_{\omega} \displaystyle\sum_{a \in A}a.n_a + \sum_{j \in J} y_j +\no= n$ is symmetrical to the previous one. 
    \item The final case whose proof we present in detail is when  $$m = \yes + \displaystyle\sum_{a \in A} a.m_a +\sum_{i \in I} x_i \simeq_{\omega} \sum_{b \in B} b.n_b + \sum_{j \in J} y_j ~[+\yes] ~[+\no]= n~,$$ where each side is in reduced normal form. To deal with this case,  we note, first of all, that by mimicking the argument in the first case of the proof, we can prove that $\mathcal{E}_{\omega,f}' \vdash n = \yes +\displaystyle\sum_{b \in B'} b.n_b' + \sum_{j \in J} y_j $, where now each $n_b'$ is $\yes$-free. By the same argument as for the verdict equivalence case (Proof of Theorem \ref{thm:CompOpenFin}) and by defining the appropriate substitutions $\sigma$  we can infer that $A = B'$ and $\{x_i \mid i \in I \} =\{y_j \mid j \in J \}$. In other words, we have:

    $$m = \yes + \displaystyle\sum_{a \in A} a.m_a +\sum_{i \in I} x_i \simeq_{\omega} \yes + \displaystyle\sum_{a \in A} a.n_a'+\sum_{i \in I} x_i ~~,$$ 

where $m$ and $n$ are in finite-action-set reduced normal form for open terms. It remains to show that under every substitution $\sigma$ we have that $\sigma(m_a) \simeq_{\omega} \sigma(n_a')$ so that we can apply our induction hypothesis and complete the proof.

Towards a contradiction assume that this is not the case. Therefore there exists a substitution $\sigma$ for which there is at least one infinite trace $s$ such that, without loss of generality, $s \in L_r(\sigma(m_a))\cdot \Act^{\omega}$ but $s \not\in L_r(\sigma(n_a'))\cdot \Act^{\omega}$ or $s \in L_a(\sigma(m_a))\cdot \Act^{\omega}$ but $s \not\in L_a(\sigma(n_a'))\cdot \Act^{\omega}$. We examine first the case of the rejection sets. Since $\sigma(m_a)$ rejects the infinite trace $s$, there is some finite prefix $s_0$ of $s$ that is rejected by $\sigma(m_a)$. Note that $\sigma(m_a)$ will also reject all the finite prefixes of $s$ that extend $s_0$. On the other hand, $\sigma(n'_a)$ does not reject any of those because it does not reject $s$.

As we saw in the proof of Theorem \ref{thm:CompOpenFin} this substitution and any such trace $s_0$ can be modified to a new substitution $\sigma'$ such that $\sigma'(m) \not\simeq \sigma'(n)$ and consequently $m$ is not verdict equivalent to $n$. Specifically from the proof of Theorem \ref{thm:CompOpenFin} we have that: 
\begin{itemize}
\item Under the substitution $\sigma'$, all variables except $x$ are mapped to $\vend$. 
\item $m_a \xrightarrow[]{s'} x + m_a'$ for some $m_a'$ and a trace $s'$ that is a prefix of $s_0$.
\item $n_a' \not\xrightarrow[]{s'} x + n_a''$ for any $n_a''$
\item The variable $x$ is mapped to $s_b.\no$  for a trace $s_b$ such that $m$ rejects the trace  $as's_b$, but $n$ does not. 
\end{itemize}

By Lemma \ref{lem:verdictOmegaVerdTracesFin} we have that the only way $m$ can be $\omega$-verdict equivalent to $n$ is if the number of traces they disagree on, under any substitution (including $\sigma'$), is finite. Since monitor $m$ is $\omega$-verdict equivalent to $n$, both monitors must disagree on finitely many extensions of $as's_b$. This however can be done only if $m_a$ and $n_a'$ also disagree on finitely many extensions of $s's_b$. This is because we have seen that under $\sigma'$, only the variable $x$ can contribute to the rejections sets of the monitors and it does so by being mapped to $s_b.\no$. However, as $s_b$ is not a prefix of $as's_b$ we know that also none of its extensions are prefixes of $as's_b$. Therefore the rejection of $s_b$ does not cause the rejection of any of the prefixes and extensions of $as's_b$. This implies that the infinite trace $s$ is only rejected by $\sigma'(m_a)$ but not $\sigma'(n_a')$, which implies that the monitors $m_a$ and $n_a'$ still disagree on infinitely many extensions of $s_0$ under the new substitution $\sigma'$ 
%and therefore infinitely many extensions of $s'.s_b$ 
which is a contradiction.

It is now easy to see that for each $a \in A$ and for each substitution $\sigma$ we have that $L_r(a.\sigma(m_a))\cdot \Act^{\omega} =  L_r(a.\sigma(n_a'))\cdot \Act^{\omega}$ which implies $L_r(\sigma(m_a))\cdot \Act^{\omega} =  L_r(\sigma(n_a'))\cdot \Act^{\omega}$. It remains to see that $L_a(\sigma(m_a))\cdot \Act^{\omega} =  L_a(\sigma(n_a'))\cdot \Act^{\omega}$.

To this end, assume, towards a contradiction, that there exist a substitution $\sigma$ and an infinite trace $s$ such that $s \in L_a(\sigma(m_a))\cdot \Act^{\omega}$ but $s \not\in L_a(\sigma(n_a'))\cdot \Act^{\omega}$. Following the argument for the rejection sets, we can infer that there is a finite trace $s_0$ accepted by $\sigma(m_a)$ but not by $\sigma(n_a')$. Again by using the proof of Theorem \ref{thm:CompOpenFin} , we can transform $\sigma$ into a $\sigma'$ that causes a disagreement over the rejection of a trace $s_0'$ for $\sigma(m_a)$ and $\sigma(n_a')$ i.e. $s_0' \in L_r(\sigma'(m_a))$ but $s_0' \not\in L_r(\sigma'(n_a'))$. This, in turn, means we can apply the same reasoning as before for the rejection of a trace to reach a contradiction, namely that $m$ and $n$ are not $\omega$-verdict equivalent.

We can therefore conclude that $\sigma(m_a) \simeq_{\omega} \sigma(n_a')$ under any substitution $\sigma$ and therefore we can apply our induction hypothesis to obtain $\mathcal{E}_{\omega,f}' \vdash m _a = n_a'$. Using the congruence rules, we have $\mathcal{E}_{\omega,f}' \vdash m=n$, and we are done. 
\qedhere
\end{itemize} \end{proof}
Table \ref{tab:allAxiomSystems} summarizes the equational axiom systems we have obtained.
\begin{table}
\rule{\textwidth}{0.2mm} 
\begin{minipage}{0.5\textwidth}
\begin{equation*}
\begin{split}
& \textbf{(A1)}~ x + y  = y + x \\
& \textbf{(A2)}~ x + (y + z)  = (x+y)+z \\
& \textbf{(A3)} ~x + x  = x \\
& \textbf{(A4)} ~x + \vend  = x 
\end{split}
\end{equation*}
\end{minipage}
\begin{minipage}{0.46\textwidth}
\begin{equation*}
\begin{split}
& \mathbf{(E_a)} ~a.\vend = \vend ~(a \in \Act)\\
& \mathbf{(Y_a)} ~ \yes = \yes + a.\yes ~(a \in \Act) \\
& \mathbf{(N_a)} ~\no = \no +  a.\no ~(a \in \Act) \\
& \mathbf{(D_a)} ~a.(x + y)  = a.x +a.y ~(a \in \Act)
\end{split}
\end{equation*}
\end{minipage}
\centering \vspace{3mm} 

The axioms of $\mathcal{E}_v$, which are ground complete for $\simeq$ (Theorem~\ref{thm:GrCompFin}).
\rule{\textwidth}{0.2mm} 
\begin{minipage}{0.5\textwidth}
\begin{equation*}
\begin{split}
& \mathbf{(Y_{\omega})}~\yes=\displaystyle\sum_{a\in \Act} a.\yes
\end{split}
\end{equation*}
\end{minipage}
\begin{minipage}{0.46\textwidth}
\begin{equation*}
\begin{split}
& \mathbf{(N_{\omega})}~ \no=\displaystyle\sum_{a\in \Act} a.\no
\end{split}
\end{equation*}
\end{minipage} 
\vspace{3mm} 

The axiom system $\mathcal{E}_{\omega} = \mathcal{E}_v \cup \{Y_{\omega},N_{\omega}\}$ is ground complete for $\simeq_{\omega}$ when $\Act$ is finite (Theorem~\ref{thm:GrCompOmega}).
\rule{\textwidth}{0.2mm} 

$\mathbf{(O1)}~\yes + \no = \yes + \no + x$
\vspace{3mm} 

The axiom system $\mathcal{E}_{v}' =  \mathcal{E}_v \cup \{O1\}$ is complete for $\simeq$ when $\Act$ is infinite (Theorem~\ref{thm:CompOpenInf}).
\rule{\textwidth}{0.2mm}

 $\mathcal{O} = \{ O2_{s,k} \mid s \in \Act^*, k \geq 0 \}$ where  \vspace{2mm}

$(\mathbf{O2_{s,k}})~ x + s.x + \overline{s}^{(k)}(\yes + \no)  = x + \overline{s}^{(k)}(\yes + \no) $
\vspace{3mm} 

The axiom system $\mathcal{E}_{\textit{v,f}}' = \mathcal{E}_v' \cup \mathcal{O}$ is complete for $\simeq$ when $\Act$ is finite and $|\Act|\geq 2$ (Theorem~\ref{thm:CompOpenFin}).
\rule{\textwidth}{0.2mm}

 $(\mathbf{V_1}) ~ a.x + x = x $
\vspace{3mm} 

The axiom system $\mathcal{E}_{v,1}' = \mathcal{E}_v' \cup \{V_1\}$ is  complete for $\simeq$ when $|\Act|=1$ (Theorem~\ref{thm:CompOpenUnary}).
\rule{\textwidth}{0.2mm}

 $(\mathbf{V1_{\omega}}) ~ x = a.x $
\vspace{3mm} 

The axiom system $\mathcal{E}_{\omega,1}' = \{A1,\ldots,A4, V1_{\omega}, O1\}$ is complete for $\simeq_{\omega}$ when $|\Act|=1$ (Theorem~\ref{thm:compOpenOmegaUnary}).
\rule{\textwidth}{0.2mm}
The axiom system $\mathcal{E}_{\omega,f}' = \mathcal{E}_{\omega} \cup \mathcal{E}_{\textit{v,f}}'$ is complete for $\simeq_{\omega}$ when $\Act$ is finite and $|\Act|\geq 2$ (Theorem~\ref{thm:CompOpenOmega}).
\rule{\textwidth}{0.2mm}

\caption{Our axiom systems}\label{tab:allAxiomSystems}
\end{table}
%%%%%%%%%%%%%%%%%%%%%%%%%%%%%%%

\section{A non-finite-axiomatizability result}\label{sec:neg}
% !TeX root = axioms_Monitors_VeEq.tex
%% Last modified: apr 19 10:51:37 GMT 2021
%% Last spell checked: 

Observe that the family of axioms $\mathcal{O} = \{ O2_{s,k} \mid s \in Act^*, k \geq 0 \}$, which is included in $\mathcal{E}_{\textit{v,f}}'$, is infinite. Thus it is natural to wonder whether verdict equivalence has a finite equational axiomatization over \MonF. In the remainder of this section, we will provide a negative answer to that question by showing that no finite subset of $\mathcal{E}_{\textit{v,f}}'$ is enough to prove all the equations in $\mathcal{O}$.

%The only case where this is not a problem is the case where $Act$ is a singleton alphabet $\{a\}$, as the equational basis for that case does not involve the family $\mathcal{O}$. When we have at least two actions available things change. Could it be that not all of these equations are necessary? This section is dedicated to proving that no finite subset of these equations is enough to prove them all. Specifically we will show that for an arbitrary finite subset of the equations in $O2_{s,k}$ there is always an infinite number of these equations that we cannot prove.

Intuitively, the proof of the above claim proceeds as follows. Let $\mathcal{E}$ be an arbitrary finite subset of $\mathcal{E}_{\textit{v,f}}'$. First of all, we isolate a property of equations that is satisfied by all the equations that are provable from $\mathcal{E}$. We then show that there are equations in the family $\mathcal{O}$ that do not have the given property. This means that those equations are not provable from $\mathcal{E}$ and, therefore, that $\mathcal{E}$ cannot be complete for verdict equivalence.

\paragraph{An arbitrary finite axiom set vs.~a finite subset of $\mathcal{E}_{\textit{v,f}}'$}

In Section \ref{sect:ComplOpenFin}, in Theorem \ref{thm:CompOpenFin}, we proved that $\mathcal{E}_{\textit{v,f}}'$ is complete for open terms over a finite action set modulo verdict equivalence. Therefore, without loss of generality, we can assume that this basis is in fact a subset of the equations in $\mathcal{E}_{\textit{v,f}}'$. To see this, consider any sound equation that could be involved in an arbitrary axiom set. Since $\mathcal{E}_{\textit{v,f}}'$ is complete this equation is derivable from it. In addition, since every proof is finite, there is a finite number of axioms of $\mathcal{E}_{\textit{v,f}}'$ involved in this proof. Therefore, any finite family of equations is derivable from a finite subset of the equations in $\mathcal{E}_{\textit{v,f}}'$. This means that if another finite family of equations was complete, there would also be a finite subset of equations from $\mathcal{E}_{\textit{v,f}}'$ which would also be complete. From now on, when considering a finite equational basis we will always mean a subset of the equations in $\mathcal{E}_{\textit{v,f}}'$.

We remind our readers that we assume that all axiom systems that we
consider are closed under symmetry. This preserves finiteness and
allows us to simplify our arguments, since the symmetry rule does not
need to be used in equational proofs.

\begin{definition}[\textbf{Notation}]
For a finite, non empty set of equations $\mathcal{E}$ we denote as $\depth(\mathcal{E})$ the quantity: 
$$max\{ \depth(m) \mid m = n \in \mathcal{E}  \}.$$
\end{definition}
The depth of an axiom system turns out to be a very important aspect of it when proving open equations.  We refer the reader to all the axioms we have defined so far (Figure \ref{tab:allAxiomSystems}) and particularly to the family $\mathcal{O}$. Take an instance of the family of equations $\mathcal{O}$, namely 

$$x + a^k.x + \overline{a^k}^{3}(\yes+\no) \simeq x + \overline{a^k}^{3}(\yes+\no) ~,$$ for some $k$. What we will focus on for equations like this one is the fact that every trace starting with $s^k$ followed by any trace of length larger than $3k +1$ (which is the depth of this equation), is both accepted and rejected by both sides of the equation for any closed substitution. This fact is exactly the intuition behind the property that we will use. We now proceed to formulate this property formally: 
\begin{lemma}\label{lem:property}
Let $\mathcal{E}$ be a finite subset of $\mathcal{E}_{\textit{v,f}}'$ and let $m = n$ be an equation in $\mathcal{E}$. Assume that for some string $s$:
\begin{itemize}
\item $m \xrightarrow[]{s} m' + x$, for some monitor $m'$ and variable $x$ and
\item $n \not\xrightarrow[]{s} n' + x$ for any $n'$.
\end{itemize} 
Then, for every trace of the form $s.s'$ where $|s'| \geq depth (\mathcal{E})$, we have that $ss'\in L_a(\sigma(m)) $ and $ss'\in L_r(\sigma(m))$ for every substitution $\sigma$. 
\end{lemma} 

\begin{proof}

It suffices to examine each member of $\mathcal{E}_{\textit{v,f}}'$ separately. 
%If the statement holds for every axiom individually then by combining them into finite sets will only cause $\depth(\mathcal{E})$ to increase which means that the statement will hold as well. 

\begin{itemize}
\item Each axiom in $\mathcal{E}_v$ does not have any one-sided occurrence of a variable as the ones stated and therefore the lemma holds vacuously.
\item For the axiom $O1$ we have that both sides accept and reject all traces for each $\sigma$ and therefore the claim follows trivially.
\item We are left to discuss the family of equations $\mathcal{O}$. Let us select an arbitrary member of this family, i.e. for some $s_0 \in Act^*$ and some $k \geq 0$, the equation $$x + s_0.x + \overline{s_0}^{(k)}(\yes+\no) =x + \overline{s_0}^{(k)}(\yes+\no)~~.$$

We see that the depth of $x$ is $1$, the depth of $s_0.x$ is $|s_0| + 1$ and the depth of the term $\overline{s_0}^k(\yes+\no)$ is $(k+1)|s_0| +1$ (which follows by the definition of the term $\overline{s}^k(m)$). We can also see that the term $\overline{s_0}^k(\yes +\no)$ accepts and rejects all traces of the form $s_0s'$, where the length of $s'$ is strictly bigger than $(k-1)|s_0|$, which is enough for the statement to hold.

\end{itemize} \end{proof}

Now that we have defined the property we were looking for over a finite subset $\mathcal{E}$ of $\mathcal{E}_{\textit{v,f}}'$, we proceed to show that the property itself is preserved by equational proofs from $\mathcal{E}$.
\begin{theorem}\label{thm:nonFin}

Let $\mathcal{E}$ be a finite subset of $\mathcal{E}_{\textit{v,f}}'$ and let $m = n$ be an equation such that $\mathcal{E} \vdash m = n$.  Assume that:
\begin{itemize}
\item $m \xrightarrow[]{s} m' + x$ for some string $s$, monitor $m'$ and variable $x$ and
\item $n \not\xrightarrow[]{s} n' + x$ for any $n'$.
\end{itemize} 
Then, for every trace of the form $s.s'$ where $|s'| \geq depth (\mathcal{E})$, we have that $ss'\in L_a(\sigma(m)) $ and $ss'\in L_r(\sigma(m))$ for every substitution $\sigma$. 
\end{theorem}

\begin{proof}

We will use induction over the length of the proof that results in an arbitrary equation $m=n$. Our base case is a proof of length one, where the the only equations we can prove are the axioms themselves and therefore the property holds by Lemma \ref{lem:property}. 

Assume now we have shown that all proofs of length up to $\ell$ preserve the property. We will show that proofs of length up to $\ell+1$ do so as well.
The final step of a proof can be performed by applying: 
\begin{itemize}
\item The congruence rule for $+$,
\item The congruence rule for action prefixing $a.\_$,
\item A variable substitution (for an open substitution $\sigma$), or
\item Transitivity. 
%\item Axiom application.
\end{itemize}
Note here that, as we mentioned earlier, the axiom system $\mathcal{E}_{\textit{v,f}}'$ is closed with respect to symmetry and therefore there is no need to use the symmetry rule in proofs.  We proceed by considering each of the above-mentioned proof steps. 
\begin{itemize}
\item The congruence rule for $+$ must be applied as so: Assume two equations $m_1 = n_1$ and $m_2 = n_2$, two already proven equations for which the statement of the theorem holds (inductive hypothesis). By applying the congruence rule for $+$ we have proven the equation $m = m_1 + m_2 = n_1+n_2 = n$. Assume that $m \xrightarrow[]{s} m' + x$ for some string $s$, monitor $m'$ and variable $x$ and $n \not\xrightarrow[]{s} n' + x$ for any $n'$. By the operational semantics of \MonF~we have that either $m_1 \xrightarrow[]{s} m' + x$ or $m_2 \xrightarrow[]{s} m' + x$. Without loss of generality assume $m_1 \xrightarrow[]{s} m' + x$.  Moreover we have that $n_1 \not\xrightarrow[]{s} n_1' + x$ for any $n_1'$ since $n \not\xrightarrow[]{s} n' + x$ for any $n'$. By inductive hypothesis then for every trace of the form $s.s'$ where $|s'| \geq depth (\mathcal{E})$ we have that $s.s'\in L_a(\sigma(m_1) )$ and $s.s'\in L_r(\sigma(m_1))$ for every substitution $\sigma$. This in turn implies that $s.s'\in L_a(\sigma(m)) $ and $s.s'\in L_r(\sigma(m))$ for every substitution $\sigma$ and we are done.

\item We now consider the case of applying the congruence rule for action prefixing. Assume a proven equation $m_0 = n_0$ on which we apply the axiom prefixing congruence rule for an action $a \in Act$, that is, $m = a.m_0 = a.\no = n$. Assume now that $m \xrightarrow[]{s} m_s + x$ for some string $s$, monitor $m_s$ and variable $x$ and $n \not\xrightarrow[]{s} n_s + x$ for any $n_s$. Since $m = a.m_0$, it follows that $s = as_0$ and $m_0 \xrightarrow[]{s_0} m_0' + x$ for some $m_0'$ and $n_0 \not\xrightarrow[]{s_0} n_0' + x$ for any $n_0'$. Therefore by inductive hypothesis we have that all traces of the form $s_0.s'$ where $|s'| \geq \depth(\mathcal{E})$ are accepted and rejected by $m_0$ under any substitution. Consequently all traces of the form $as_0.s' = ss'$ are both accepted and rejected by $m$ under any substitution and we are done.

%
% By induction hypothesis for every $s'$ where $|s'| \geq k$ we have that $s.s'\in L_a(\sigma(m)) $ and $s.s'\in L_a(\sigma(m))$ for every substitution $\sigma$. After applying the axion prefixing rule we have that the variable now occurs after the trace $a.s$ and additionally for every trace $s'$ with $|s'| \geq k$ we have that both $a.m$ and $a.n$ can perform the transition $ \xrightarrow[]{a.s.s'} \yes +\no $  which means that the property is preserved. 

\item Consider now variable substitution. Note that we will consider open substitutions, in order to capture the more general case. The case of closed substitutions is of course trivial as after one of them is applied there are no variable occurrences left in any equation and therefore the result holds vacuously. We have now that $\mathcal{E} \vdash m'=n'$ for some open monitors $m'$ and $n'$ and that we apply the open substitution $\sigma_0$ in order to prove the open equation $m = \sigma_0(m') = \sigma_0(n') = n$. Assume now that $\sigma_0(m) \xrightarrow[]{s} m_s + x$ for some string $s$, monitor $m_s$ and variable $x$ and $\sigma_0(n) \not\xrightarrow[]{s} n_s + x$ for any $n_s$. We can easily see that every such one-sided occurrence of a variable in the new equation must have resulted from a one-sided variable occurrence in $m' = n'$. This is because if there were no one-sided variable occurrences in the old equation, then under no substitution could one have introduced a variable in only one side without also introducing it on the other side. This means that there exists some variable $y$ (which could be the same as $x$) such that $m' \xrightarrow[]{s_0} m_{s_0}' + y$ for some string $s_0$ where $s_0$ a prefix of $s$, monitor $m_{s_0}'$ and variable $y$ and $n' \not\xrightarrow[]{s_0} n_{s_0}' + y$ for any $n_{s_0}'$. The reason why $s_0$ must be a prefix of $s$ is that an open substitution can only expand the traces that lead to a variable occurrence in the original term. By applying our inductive hypothesis on $m' = n'$, we have that both $m'$ and $n'$ must accept and reject all traces of the form $s_0.s'$ where $|s'| \geq depth (\mathcal{E})$ under any substitution $\sigma$. 
This, in turn, implies that $\sigma_0(m') = m$ accepts and rejects traces of the form $ss'$ under any closed substitution $\sigma$. In fact $\sigma(\sigma_0(m')) = \sigma_0(\sigma_0(m'))$ which means that $m$ and $n$ reject the traces of the form $s_0.s$ as well. Since $s_0$ is a prefix of $s$ we have that for every extension of $s$ of length at least $\depth(\mathcal{E})$ there exists an extension of $s_0$ of length at least $\depth(\mathcal{E})$ that is a prefix of it. Since all traces $s_0.s'$ of this length are both accepted and rejected under any substitution, the same applies for the traces $s.s'$ and we are done.

\item The case of transitivity is also straightforward though the following inductive argument. We start by $\mathcal{E} \vdash m = m'$ and $\mathcal{E} \vdash m' = n$ and we apply the transitivity rule to prove $m =n$. Assume that $m \xrightarrow[]{s} m_s + x$ for some trace $s$, variable $x$ and monitor $m_s$, while $n \not\xrightarrow[]{s} n_s + x$ for any $n_s$. We have that either:  $m' \xrightarrow[]{s} m_s' + x$ for some $m_s'$ or $m' \not\xrightarrow[]{s} m_s' + x$. In the first case we have that the equation $m'=n$ which has already been proven by $\mathcal{E}$ satisfies the premises of the theorem and therefore by induction hypothesis all traces of the form $s.s'$ where $|s'| \geq \depth(\mathcal{E})$ are both accepted and rejected by both $m'$ and $n$. Since $n \simeq m$ by the soundness of $\mathcal{E}_{\textit{v,f}}'$ and thus $\mathcal{E}$, we have that $m$ also accept and rejects all of these traces and we are done. In the second case and via a similar argument we have the same result.

%
%\item We finally have to analyze the case of axiom application to some sub-term of a proven equation $m=n$. This is also fairly straightforward as we have that since the axioms are sound, when applied the do not alter the acceptance and rejection sets of the equation in which they are instantiated. Therefore if the one sided occurrence of a variable existed in an equation the necessary traces stated by the property also are there. If a one sided occurrence of a variable is created because of the axiom application we also have that the necessary traces must have been accepted and rejected before the application of the axiom so that the axiom could be applied. 

\end{itemize}
This concludes the case analysis for our inductive proof and we are done. \end{proof}

As we can see, if we start from any finite subset $\mathcal{E}$ of $\mathcal{E}_{\textit{v,f}}'$, we are bound to only prove equations that have the property in the statement of Theorem \ref{thm:nonFin}. We now argue that for each $\mathcal{E}$ there will always exist sound equations in $\mathcal{E}_{\textit{v,f}}'$ that do not satisfy the above property and therefore the axiom set $\mathcal{E}$ is not enough to prove them. 
\begin{lemma}
Let $\mathcal{E}$ be a finite subset of $\mathcal{E}_{\textit{v,f}}'$. There exists a sound  equation $m = n$ in $\mathcal{O}$ such that $m \xrightarrow[]{s} m' + x$ for some string $s$, monitor $m'$ and variable $x$ and $n \not\xrightarrow[]{s} n' + x$ for any $n'$ and there is at least one trace of the form $s.s'$ where $|s'| \geq \depth(\mathcal{E})$ and $s.s'\not\in L_a(\sigma(m)) $ and $s.s' \not\in L_a(\sigma(m))$ for the one substitution $\sigma_{\vend} = \vend$, for every $x$. 
\end{lemma}

\begin{proof}
It suffices to give an example from the members of the family $\mathcal{O}$. Namely we consider the equation: $$ x + a^n.x + \overline{(a^n)}^3(\yes+\no)= x + \overline{(a^n)}^3(\yes+\no) ~,$$ where $n > \depth(\mathcal{E})$.

We can clearly see that first of all the occurrence of $x$ after the trace $a^n$ is one-sided in the left hand side of the equation.  However there is a substitution (namely $\sigma(x)  = \vend $) under which the trace $a^{2n +1}$ is neither accepted nor rejected by the two monitors even though the length of $a^{(n+1)}$ is strictly larger than $\depth(\mathcal{E})$. \end{proof}

\begin{theorem}\label{thm:nonfinfinal}
There is no finite complete set of axioms for verdict equivalence over \MonF~over a finite, non-unary set of actions.
\end{theorem}

\begin{proof}
Let $\mathcal{E}$ be a finite subset $\mathcal{E}_{\textit{v,f}}'$. Then, by the above lemma, $\mathcal{E}$ cannot prove the sound equation $$ x + a^n.x + \overline{(a^n)}^3(\yes+\no)= x + \overline{(a^n)}^3(\yes+\no) ~,$$ for $n > \depth(\mathcal{E})$ and we are done. \end{proof}

\section{Conclusions}\label{Sect:conclusions}  
% !TeX root = axioms_Monitors_VeEq.tex

In this article, we have studied the equational theory of recursion-free, regular monitors from~\cite{AcetoPOPL19,operGuidetoMon,FraAI17} modulo two natural notions of monitor equivalence, namely verdict and $\omega$-verdict equivalence.
 We have provided complete axiomatizations for those equivalences over closed and open terms.
  The axiomatizations over closed terms are finite when so is the set of actions monitors can process.
  On the other hand, even when the set of actions is finite, whether those equivalences have finite bases over open terms depends on the cardinality  of the action set.
  For instance, we have shown that verdict equivalence has no finite equational axiomatization when the set of actions contains at least two actions.

Since verdict and $\omega$-verdict equivalence are trace-based behavioral equivalences, our axiomatizations, which are summarized in  Table~\ref{tab:allAxiomSystems}, share a number of equations with those for trace and completed trace equivalence over BCCSP~\cite{Gla01} and for equality of regular expressions~\cite{JHConway71,Kozen94,Salomaa66}. However, the presence of the $\yes$, $\no$ and $\vend$ verdicts yields a number of novelties and technical complications, which are most evident in the axiomatization results over open terms and in the negative result we present in Section~\ref{sec:neg}. By way of example, we remark here that, as mentioned in~\cite{ChenFLN08}, trace and completed trace equivalence are finitely based over BCCSP when the set of actions is finite, unlike the notions we study in this paper over monitors. Moreover, unlike the one given in this paper, proofs of non-finite-axiomatizability results for regular expressions rely on families of equations that exploit the interplay between Kleene star and concatenation, such as 
\[
a^* = (a^n)^*(1+a+\cdots+a^{n-1}) \quad (n>0) . 
\]
See, for instance,~\cite{AcetoFI98,JHConway71,redko_reg_lang}.

The results presented in this article deal with a minimal language for
monitors that is mainly of theoretical interest and set the stage for
further research.  An interesting and natural avenue for future work
is to study the complexity of the equational theory of verdict and
$\omega$-verdict equivalence. Moreover, one could investigate
axiomatizations of those behavioral equivalences over extensions of
recursion-free monitors with the parallel operators considered
in~\cite{AcetoPOPL19} and/or with recursion \cite{FraAI17}. As shown
in~\cite{AcetoPOPL19}(Proposition~3.8), every `reactive parallel
monitor' is verdict equivalent to a regular one. This opens the
tantalizing possibility that verdict equivalence affords an elegant
equational axiomatization over such monitors. However, the proof of
Proposition~3.8 in~\citep{AcetoPOPL19} relies on a non-trivial
automata-theoretic construction, which would have to be simulated
equationally to transform `reactive parallel monitors' into regular
ones. We leave this interesting problem for further study.

\bibliographystyle{model5-names}
%\biboptions{authoryear}
\bibliography{mybibfile}

\end{document}